\documentclass[11pt]{article}
\usepackage[a4paper,centering,scale=0.7]{geometry}
\usepackage{amsmath,amsfonts,amsthm,amssymb}
\usepackage{booktabs}
\usepackage{graphicx}
\usepackage{color}

\newtheorem{prop}{Proposition}[section]

\newtheorem*{hyp}{Hypothesis (H)}
\theoremstyle{remark}
\newtheorem{rmk}[prop]{Remark}

\newcommand{\E}{\mathbb{E}}
\newcommand{\cF}{\mathcal{F}}
\renewcommand{\P}{\mathbb{P}}
\newcommand{\Q}{\mathbb{Q}}
\newcommand{\erre}{\mathbb{R}}

\newcommand{\ds}{\displaystyle}

\title{Nonparametric estimates of pricing functionals}
\author{Carlo Marinelli\thanks{Department of Mathematics, University
    College London, Gower Street, London WC1E 6BT, UK.} \and Stefano
  d'Addona\thanks{Corresponding Author. Department of Political
    Science, University of Rome 3, Via G. Chiabrera, 199, I-00145
    Rome, Italy. Tel.~+39-06-5733-5331; e-mail
    \texttt{daddona@uniroma3.it}}}
\date{June 8, 2017}

\begin{document}
\maketitle

\begin{abstract}
  We analyze the empirical performance of several non-parametric
  estimators of the pricing functional for European options, using
  historical put and call prices on the S\&P500 during the year
  2012. Two main families of estimators are considered, obtained by
  estimating the pricing functional directly, and by estimating the
  (Black-Scholes) implied volatility surface, respectively. In each
  case simple estimators based on linear interpolation are
  constructed, as well as more sophisticated ones based on smoothing
  kernels, \`a la Nadaraya-Watson. The results based on the analysis
  of the empirical pricing errors in an extensive out-of-sample study
  indicate that a simple approach based on the Black-Scholes formula
  coupled with linear interpolation of the volatility surface
  outperforms, both in accuracy and computational speed, all other
  methods.

\medskip


\medskip

\end{abstract}

\section{Introduction}
\label{sec:intro}
The purpose of this work is to analyze the empirical performance of
some non-parametric and semi-parametric estimators of pricing
functional, with particular emphasis on the simplest contingent
claims, i.e. European put and call options. The well-known idea
underlying non-parametric estimation is the following: assume that the
price of a contingent claim of a certain type can be written as
$\pi(y^1,\ldots,y^m)$, where $\pi:\erre^m \to \erre$ is a function to
be estimated, and $y^1,\ldots,y^m$ are observable parameters. Given a
sample
\[
(y_k)_{1\leq k \leq N} = \bigl(y^1_k,\ldots,y^m_k\bigr)_{1\leq k \leq N},
\qquad (\pi_k)_{1\leq k \leq N} = 
\bigl( \pi(y^1_k,\ldots,y^m_k) \bigr)_{1\leq k \leq N}
\]
and $y_0 \in \erre^m$, one can estimate $\pi(y_0)$ by linear
interpolation (for instance) of the function
$\tilde{\pi}:\erre^m \to \erre$ defined as
\[
\tilde{\pi}(y) := \sum_{k=1}^N \pi_k \mathbf{1}_{\{y_k\}}(y).
\]
This procedure is well understood if $y_0$ belongs to the convex hull
of $(y_i)_{1\leq i \leq N}$ (see \S\ref{ssec:li} for more details). 
There exist of course many other procedures
based on nonlinear interpolation, rather than linear, that may be
preferable in certain situations, for instance to obtain estimators
that are more regular than just continuous. The same problem has also
been studied from a statistical perspective, leading to the large
literature on non-parametric estimation of regression
functions\footnote{The even larger literature on non-parametric density
  estimation treats a less general but very related problem.} (see,
e.g.,~\cite{Bierens} and references therein). One of the most popular
estimators in this literature is the so-called Nadaraya-Watson
estimator (see \cite{Nadaraya, Watson}), that is
\[
\hat{\pi}_\varepsilon(y) = 
\frac{\ds \sum_{k=1}^m \rho_\varepsilon(y-y_k) \pi_k}%
     {\ds \sum_{k=1}^m \rho_\varepsilon(y-y_k)},
\]
where $\rho:\erre^m \to \erre$ is a strictly positive continuous
radial function with integral equal to $1$ and
$\rho_\varepsilon := \varepsilon^{-m}\rho(\cdot \, \varepsilon^{-1})$,
or a slight generalization thereof.

Other functions of the parameters $(y^1,\ldots,y^m)$ can clearly be
estimated in the same manner. In particular, if the parameters are the
usual inputs in the Black-Scholes pricing formula (except the
volatility), one can estimate the implied volatility, which can then be
fed back to the Black-Scholes formula to produce further estimates of
the pricing functional.

One of the first works applying non-parametric regression \`a la
Nadaraya-Watson in pricing problems is \cite{AL}, where the authors,
among other things, report impressive empirical results on the
precision of semi-parametric estimators applied to pricing European
call options on the S\&P500. To date, a large literature deals with
this approach to pricing %\footnote{We do not know, however, whether
% pricing by non-parametric regression has (had) any interest among
% practitioners.}  
-- see, e.g.,~\cite{GHS} and references therein, as
well as \cite{ChernovGhysels:00, RosenbergEngle:02} for related
ideas.

Our purpose is to understand how different (but related)
non-parametric approaches perform in terms of pricing precision, also
in comparison with non-trivial fully parametric alternatives. A
natural question, of clear relevance for practical purposes, is
whether Nadaraya-Watson kernel estimation produces better estimates
than elementary linear interpolation. In order to avoid rather heavy
Markovianity and stationarity assumptions on the price process of the
underlying index, we use observed option prices in a given day to
estimate (unobserved!) option prices in the \emph{same} day. It turns
out that, from the empirical point of view, there does not seem to be
any advantage connected with Nadaraya-Watson kernel estimators, which
perform rather poorly in comparison with simple linear-interpolation
estimators. The latter also consistently outperform a benchmark
parametric estimator based on the (skewed) Variance-Gamma process (see
\cite{MCC}).

The focus of this work is somewhat different from that of \cite{AL},
whose main aim is to estimate the so-called state-price
density. However, such estimates are obtained by the second derivative
with respect to the strike price of (estimates of) the pricing
functional for European call options, and their main (practical)
application seems to be pricing anyway. On the other hand, the main
terms of comparison used in \cite{AL} are rather involved methods
based on neural networks and on implied binomial trees, and empirical
tests are done in a very different way: observed option prices over a
period of nine months are aggregated to construct a Nadaraya-Watson
estimator of the implied volatility surface, seen as a function of
underlying's futures price, strike, and time to maturity. This
estimator is then used to \emph{forecast} prices of European call
options (with strike prices within a certain range) at five later dates,
namely from one to twenty days in the future.

We have not been able to find in the literature neither empirical
studies where observed option prices over different dates are
\emph{not} aggregated, nor a discussion about the practical aspects of
non-parametric estimation of pricing functionals, for the mere purpose
of pricing \emph{plain vanilla} European options.\footnote{One can
  find, however, articles that use non-parametric methods to study
  related problems, without aggregating option prices across days. For
  instance, \cite{JR96} estimate the risk-neutral density and
  \cite{Beare} test the monotonicity of the pricing kernel.} Our goal
is to try and answer some basic fundamental questions in this regard,
on the basis of an extensive out-of-sample analysis.

The rest of the paper is organized as follows: in Section
\ref{sec:prelim}, after recalling some facts about yield processes
associated to dividend-paying assets that are not easily found in
standard textbooks, we prove a put-call parity identity for European
options written on a dividend-paying asset, under a conditional
uncorrelation assumption involving the asset, dividend rate, and
risk-free rate processes. These basic theoretical results are needed
because the S\&P500 pays dividends. The various estimators used in the
empirical analysis to follow are introduced in detail in Section
\ref{sec:estim}. A preliminary analysis on the data set of option
prices on the S\&P500 for the whole year 2012 is conducted in Section
\ref{sec:data}. Finally, Section \ref{sec:ea} contains a detailed
study of the empirical performance of the estimators introduced in
Section \ref{sec:estim}.

%-------------------------------------------------------------------

\section{Setting and preliminaries}
\label{sec:prelim}
Given a probability space $(\Omega,\cF,\P)$ endowed with a
(right-continuous, complete) filtration $\mathbb{F}=(\cF_t)_{0\leq t\leq T}$,
let the adapted processes $\beta$ and
$S:\Omega \times [0,T] \to \erre$ describe the price of two traded
assets: the former is just the risk-free cash account, i.e.
\[
\beta_t = \exp\biggl( \int_0^t r_s\,ds \biggr) \qquad \forall t \in [0,T],
\]
where $r$ denotes the adapted, $\P$-a.s. positive, short-rate process;
the latter is a risky asset with an associated adapted,
$\P$-a.s. positive, dividend rate process $q$, such that
\[
\exp \biggl( \int_0^T q_u\,du \biggr) \in L^\infty(\Omega,\cF_T,\P).
\]
The corresponding yield process $Y$ is defined as
\[
Y_t = S_t + \int_0^t q_sS_s\,ds \qquad \forall t \in [0,T].
\]
We assume that there exists a probability measure $\Q$, equivalent to
$\P$, such that, setting $\check{S}_t := \beta_t^{-1}S_t$, the
discounted yield process $\check{Y}$, defined by
\[
\check{Y}_t := \check{S}_t + \int_0^t q_u \check{S}_u\,du \equiv
\beta_t^{-1}S_t + \int_0^t q_u\beta_u^{-1}S_u\,du
\qquad \forall t \in [0,T],
\]
is a square integrable $\Q$-martingale. Setting
$A:=e^{\int_0^\cdot q_u\,du}$, the integration by parts formula yields
\[
\check{S}_tA_t = S_0 + \int_0^t \check{S}_{u-}\,dA_u 
+ \int_0^t A_{u-}\,d\check{S}_u + [\check{S},A]_t,
\]
where
\[
\int_0^t \check{S}_{u-}\,dA_u = \int_0^t \check{S}_{u-} A_u q_u\,du
= \int_0^t \check{S}_u A_u q_u\,du,
\]
and, by continuity of $A$,
\[
\int_0^t A_{u-}\,d\check{S}_u = \int_0^t A_u\,d\check{S}_u, 
\qquad [\check{S},A]_t=0.
\]
In particular,
\[
S_t \, \exp\biggl(\int_0^t (q_u-r_u)\,du \biggr) = 
\check{S}_tA_t = S_0 + \int_0^t A_u\,d\check{Y}_u
\]
is a (square integrable) $\Q$-martingale , because $A$ is bounded and
predictable.

As is well known, the existence of an equivalent probability measure
$\Q$ such that the discounted yield process $\check{Y}$ is a
$\Q$-martingale implies that the market defined by $\beta$ and $S$ is
free of arbitrage. In general $\Q$ is not unique, unless the market is
complete, and we assume that $\Q$ is just a fixed pricing
measure. More precisely, we assume that we are given a pricing
functional
\begin{align*}
\pi: L^\infty(\Omega,\cF_T,\P) &\to \erre\\
X &\mapsto \E_{\Q} \beta_T^{-1}X,
\end{align*}
or, equivalently,
\[
\pi(X) := \E Z_TX, \qquad Z_T := \frac{1}{\beta_T} \frac{d\Q}{d\P},
\]
where the random variable $Z_T$ often goes under the name of
stochastic discount factor.

We shall write, for compactness of notation,
\begin{gather*}
A(t,T) := \exp\biggl(\int_t^T q_u\,du\biggr),
\qquad \bar{q}(t,T) := \frac{1}{T-t} \log \E_\Q[A(t,T)\vert\cF_t],\\
B(t,T) := \exp\biggl(-\int_t^T r_u\,du\biggr),
\qquad
\bar{r}(t,T) := - \frac{1}{T-t} \log \E[B(t,T)\vert\cF_t],
\end{gather*}
where $0 \leq t \leq T$.

We are going to use several times the following assumptions on $r$,
$q$, and $S$.
\begin{hyp}
  One has, for any $0 < t \leq T$,
  \[
  \operatorname{Cov}_{\Q}\bigl(A(t,T),B(t,T) \big\vert \cF_t\bigr) = 0, 
  \qquad
  \operatorname{Cov}_{\Q}\bigl(A(t,T),B(t,T)S_T \big\vert \cF_t\bigr) = 0.
  \]
\end{hyp}
\noindent Here, for any two random variables $X$, $Y \in L^2(\Omega,\cF_T,\Q)$,
\[
\operatorname{Cov}_\Q\bigl(X,Y \big\vert \cF_t \bigr) :=
\E_\Q\bigl[ \bigl( X - \E_\Q[X\vert\cF_t] \bigr)
\bigl( Y - \E_\Q[Y\vert\cF_t] \bigr)
  \big\vert \cF_t \bigr].
\]
Hypothesis (H) is equivalent to
\begin{align*}
\E_\Q[A(t,T) B(t,T) \vert \cF_t] &= \E_\Q[A(t,T) \vert \cF_t]
\, \E_\Q[B(t,T) \vert \cF_t],\\
\E_\Q[A(t,T) B(t,T) S_T \vert \cF_t] &= \E_\Q[A(t,T) \vert \cF_t]
\, \E_\Q[B(t,T) S_T \vert \cF_t]
\end{align*}
for all $0 < t \leq T$.

\subsection{Put-call parity with dividends and applications}
To implement some estimators of the pricing functional, but also to
carry out a preliminary analysis on the raw option prices, we shall
need a version of put-call parity for European call and put
options. This will also yield a natural estimator for (a functional
of) the dividend rate process $q$.
\begin{prop}   \label{prop:parity}
  Assume that Hypothesis \emph{(H)} holds. Then, for any
  $0 \leq t \leq T$,
  \begin{equation}   \label{eq:parity}
    S_t e^{-\bar{q}(t,T)(T-t)} - K e^{-\bar{r}(t,T)(T-t)} = C_t - P_t,
  \end{equation}
  where $C_t$ and $P_t$ denote the prices at time $t$ of European call
  and put options, respectively, with maturity $T$, strike price $K$,
  and underlying price process $S$.
\end{prop}
\begin{proof}
  Multiplying the identity 
  \[
  S_T-K = (S_T-K)^+ - (K-S_T)^-
  \]
  by $\beta_T^{-1} = B(0,T) = e^{-\int_0^T r_u\,du}$ and taking
  conditional expectation yields
  \begin{multline*}
    \E_\Q\bigl[\beta_T^{-1}S_T \big\vert \cF_t \bigr] 
    - K\E_\Q\bigl[\beta_T^{-1} \big\vert \cF_t \bigr] \\ 
     = \E_\Q \bigl[\beta_T^{-1} (S_T-K)^+ \big\vert \cF_t \bigr] 
      - \E_\Q \bigl[\beta_T^{-1} (K-S_T)^- \big\vert \cF_t \bigr]
     = B(0,t)\bigl( C_t - P_t \bigr).
  \end{multline*}
  Using Hypothesis (H) and the martingale property of the discounted yield
  process associated to $S$ and $q$, one has
  \[
  \E_\Q\bigl[\beta_T^{-1}S_T \big\vert \cF_t \bigr] \,
  \E_\Q\bigl[A(0,T) \big\vert \cF_t \bigr] = 
  \E_\Q\bigl[\beta_T^{-1} A(0,T) S_T \big\vert \cF_t \bigr]
  = A(0,t) B(0,t) S_t,
  \]
  where
  $\E_\Q\bigl[A(0,T) \big\vert \cF_t \bigr] = A(0,t) \E_\Q\bigl[A(t,T)
  \big\vert \cF_t \bigr]$, hence
  \[
  \E_\Q\bigl[\beta_T^{-1}S_T \big\vert \cF_t \bigr] = S_t \, B(0,t)
  \frac{1}{\E_\Q\bigl[A(t,T) \big\vert \cF_t \bigr]}
  = S_t B(0,t) e^{-\bar{q}(t,T)(T-t)}.
  \]
  Similarly,
  \[
  \E_\Q\bigl[\beta_T^{-1} \big\vert \cF_t \bigr] = 
  B(0,t) \E_\Q\bigl[B(t,T) \big\vert \cF_t \bigr] =
  B(0,t) e^{-\bar{r}(t,T)(T-t)}.
  \]
  Collecting terms and simplifying yields
  \[
  S_t e^{-\bar{q}(t,T)(T-t)} - K e^{-\bar{r}(t,T)(T-t)} = C_t - P_t.
  \qedhere
  \]
\end{proof}

Recall that the forward price at time $t \leq T$ of a contingent claim
$X \in L^2(\Omega,\cF_T,\Q)$ is defined as
\[
f_t^X := \frac{\E_\Q[X\beta_T^{-1} \vert \cF_t]}%
             {\E_\Q[\beta_T^{-1} \vert \cF_t]}
\]
(see, e.g.,~\cite[{\S}2.4]{KS-mmf}).
Using arguments entirely analogous to those in the proof of the above
Proposition, one obtains the following ``spot-forward parity''.
\begin{prop}
  If Hypothesis \emph{(H)} holds, then
  \[
  f_t^{S_T} = S_t \, e^{(\bar{r}(t,T)-\bar{q}(t,T))(T-t)}.
  \]
\end{prop}
In particular, \eqref{eq:parity} could also be written as
\[
f^{S_T}_t e^{-\bar{r}(t,T)(T-t)} - K e^{-\bar{r}(t,T)(T-t)} = C_t -
P_t.
\]

\medskip

The put-call parity identity of Proposition \ref{prop:parity} implies
a procedure to estimate $\bar{q}(t,T)$: assuming that an estimator of
$\bar{r}(t,T)$ is available and is denoted by the same symbol, and
that prices of call and put options with \emph{the same} maturity and
strike price are observable, identity \eqref{eq:parity} yields
\begin{equation}
\label{eq:qest}
\bar{q}(t,T) = -\frac{1}{T-t} 
\log \frac{C(t,T,K) - P(t,T,K) + Ke^{-\bar{r}(t,T)(T-t)}}{S_t},
\end{equation}
where $C(t,T,K)$ and $P(t,T,K)$ denote the price at time $t$ of
European call and put options, respectively, with maturity $T$ and
strike price $K$. This estimate of $\bar{q}(t,T)$ is usually preferred
in applications to option pricing to other proxies of $q$, such as
historical estimates.

%-------------------------------------------------------------------

\section{Estimators of the pricing functional}
\label{sec:estim}
From now on we shall assume that
\begin{itemize}
\item[(a)] the filtration $\mathbb{F}$ is the (right-continuous,
  completed) filtration generated by $(S,r,q)$;
\item[(b)] the process $(S,r,q)$ is Markovian, i.e., for any
  bounded random variable $\xi$ measurable with respect to
  $\sigma\bigl( (S_s,r_s,q_s)_{t \leq s \leq T}\bigr)$, one has
  \[
  \E[\xi \vert \cF_t] = \E[\xi \vert t,S_t,r_t,q_t].
  \]
\end{itemize}
These assumptions immediately imply that there exists a function
$\pi_p:\erre^6 \to \erre$ such that
\[
\pi_p(t,S_t,r_t,q_t,K,T) = 
\beta_t \E_{\Q}\bigl[ \beta^{-1}_T (K-S_T)^+ \big\vert \cF_t
\bigr]
= \E_\Q\Bigl[ e^{-\int_t^T r_u\,du} (K-S_T)^+ \Big\vert \cF_t \Bigr]
\]
for all $t \in [0,T]$. In particular, without any further assumptions,
the pricing functional will be time-dependent, hence it is wrong, in
general, to use the price functional estimated at time $t$ to price
options at time $t+1$, or, similarly, to estimate the pricing
functional aggregating data of different dates.\footnote{For a related
  discussion, also from an economic point of view,
  cf.~\cite{RosenbergEngle:02}.} This procedure become meaningful if
we further assume that the process $(S,r,q)$ is a time-homogeneous
Markov process: in this case $\pi_p$ depends on $t$ and $T$ only
through their difference $T-t$. However, it is far from clear that a
time-homogeneity assumption on the data-generating process would be
appropriate, hence the empirical analysis of Section \ref{sec:ea} is
carried out with $t$ fixed. In \S\ref{ssec:markov} we also provide
empirical data that does not seem to support the plausibility of a
time-homogeneity assumption.

\medskip

We now proceed to describe in detail the various estimators of the
pricing functional for put options that will be tested in Section
\ref{sec:ea}. Analogous considerations, not spelled out in detail,
hold for call options and, in fact, for arbitrary European
derivatives: some care has to be taken only if the payoff function is
unbounded, in which case integrability conditions have to be assumed.

\subsection{Linear interpolator}
\label{ssec:li}
Let $(Y_k)_{1\leq k \leq N}$ and $(p_k)_{1\leq k \leq N}$ subsets of
$\erre^n$ and $\erre$, and $f: \erre^n \to \erre$ a function such that
$f(Y_k)=p_k$ for all $1 \leq k \leq N$. Denoting the convex hull of
$(Y_k)_{1\leq k \leq N}$ by $C$, the linear interpolator
$\hat{f}:C \to \erre$ is a function obtained by linear interpolation
of the function $\langle\mu_N,\mathbf{1}\rangle$, where $\mu_N$ is the
marked empirical measure
\[
\mu_N := \sum_{k=1}^N p_k \delta_{Y_k},
\]
$\delta$ stands for the Dirac measure, and
$\langle\mu_N,\mathbf{1}\rangle := \sum_{k=1}^N p_k
\mathbf{1}_{\{Y_k\}}$.

Linear interpolation here means that the
Delaunay triangulation of $(Y_k)$ is computed, and, for any
$y_0 \in C$, $\hat{f}(y_0)$ is obtained by interpolation of
$\langle\mu_N,\mathbf{1}\rangle$ on the vertices of the simplex
containing $y_0$. We recall that a triangulation of the set of points
$(Y_k)_{1\leq k \leq N}$ is a partitioning of the convex hull $C$ into
simplices whose vertices are the points $(Y_k)_{1\leq k \leq N}$. In
particular, any two simplices of such a partition either do not
intersect or share a common face. A Delaunay triangulation of
$(Y_k)_{1\leq k \leq N}$ satisfies the following additional property:
if $B$ is a ball in $\erre^n$ such that all vertices of a simplex
belong to its boundary, then the interior of $B$ does not contain any
element of $(Y_k)_{1\leq k \leq N}$. Once a Delaunay triangulation of
$(Y_k)_{1\leq k \leq N}$ has been obtained, $\hat{f}(y_0)$, with
$y_0 \in C$, is defined as follows: let $v_1,\ldots,v_d$ the vertices
of the unique simplex $V$ such that $y_0 \in V$, and write
\[
  y_0 = \sum_{j=1}^d \alpha_j v_j, \qquad \alpha_j \geq 0, \qquad \sum
  \alpha_j=1.
\]
Then we set $\hat{f}(y_0) := \sum_j \alpha_j f(v_j)$. For an extensive
treatment of these topics, see, e.g.,~\cite{PreSha}.

\medskip

Let $t>0$ be fixed. Then, as seen above, we can write
\[
\E_\Q\Bigl[ e^{-\int_t^T r_u\,du} (K-S_T)^+ \Big\vert \cF_t \Bigr] 
= \pi_p(K,T-t).
\]
Assuming that $S_t>0$, adaptedness of $S$ and Markovianity
imply
\begin{align*}
  \pi_p(K,T-t) &=
  \E_{\Q} \bigl[ e^{-\int_t^T r_u\,du} (K-S_T)^+ \big\vert \cF_t \bigr]\\
  &= S_t \E_{\Q} \bigl[ e^{-\int_t^T r_u\,du} \bigl(S_t^{-1}K - S_t^{-1}S_T \bigr)^+
  \big\vert \cF_t \bigr]\\
  &= S_t \phi(K,T-t).
\end{align*}
Denoting by $\hat{\phi}$ the linear interpolator of $\phi$, we define
the normalized linear interpolator $\hat{\pi}_p$ of $\pi_p$ by
$\hat{\pi}_p := S_t\hat{\phi}$. This is the estimator used in the
empirical analysis of Section \ref{sec:ea}.

\medskip

Note that the linear interpolator is defined only on the convex hull
$C$ of the couples $(K_k,T_k-t)_{1\leq k\leq N}$. In other words, this
estimator is unable to estimate the price of options whose parameters
$(K,T-t)$ do not belong to $C$. From a purely numerical point of view,
the estimator $\hat{\pi}_p$ could be extended outside $C$, for
instance by extrapolation. However, the estimates obtained this way
are rarely reliable. Another possibility to extend the estimator
outside $C$ is to add fictitious couples $(K_j,T_j-t)_j$ to the sample
where the value of the function $\pi_p$ is known, e.g. for $T_j-t=0$,
or for $K_j$ ``very large''. This idea, which is obviously more
meaningful than a mere numerical extrapolation, may lead, in some
situations, to relatively satisfactory results. Details are discussed
in Section \ref{sec:ea} below.

\subsection{Nadaraya-Watson estimator}
\label{ssec:NW}
Let $(Y_k,p_k)_{1 \leq k \leq N}$ and $f$ be as in the previous
subsection. Let $\rho:\erre \to \erre_+$ be an integrable function
with integral equal to $1$, and set, for any $\varepsilon>0$,
$\rho_\varepsilon(x)=\varepsilon^{-1} \rho(x \varepsilon^{-1})$.  With
a slight abuse of notation, for any
$\varepsilon = (\varepsilon_1,\ldots,\varepsilon_n) \in (0,\infty)^n$,
we define the function $\rho_\varepsilon:\erre^n \to \erre$ as
$\rho_\varepsilon = \rho_{\varepsilon_1} \otimes \cdots \otimes
\rho_{\varepsilon_n}$, i.e.
\[
\rho_\varepsilon(x) \equiv \rho_\varepsilon(x_1,\ldots,x_n)
 = \rho_{\varepsilon_1}(x_1) \, \cdots \, \rho_{\varepsilon_n}(x_n).
\]
The Nadaraya-Watson (NW) estimator $\hat{f}_\varepsilon$ of
$f:\erre^n \to \erre$, with smoothing parameters
$\varepsilon=(\varepsilon_1,\ldots,\varepsilon_n)$, based on the
observations $(Y_k,p_k)_k$, is defined as
\[
\hat{f}_\varepsilon(x) = \frac{\ds \sum_{k=1}^N p_k \rho_\varepsilon(x-Y_k)}%
{\ds \sum_{k=1}^N \rho_\varepsilon(x-Y_k)}.
\]

The value of the function $\hat{f}_\varepsilon$ at $x$ is thus a
weighted average of the observations $(p_k)_k$ of the type
\[
\hat{f}_\varepsilon(x) = \sum_{k=1}^N p_k w_k\bigl(x;(Y_k)\bigr),
\qquad w_k\bigl(x;(Y_k)\bigr) := \frac{\rho_\varepsilon(x-Y_k)}%
{\sum_k \rho_\varepsilon(x-Y_k)}.
\]
Usually $\rho$ is such that $\rho = g \circ \lvert\cdot\rvert$
(i.e. it is symmetric), with $g$ is decreasing, so that the NW
estimator effectively computes a weighted average of the observations
assigning more weight to those closer to $x$. If $\rho$ has compact
support, the average is on a finite number of points whose distance
from $x$ does not exceed a certain threshold (depending on
$\varepsilon$). The NW estimator can also be interpreted as a local
constant least square approximation of the observed values $(p_k)$,
because (assuming for simplicity $n=1$)
\[
\hat{f}_\varepsilon(x) = \arg\,\min_{\theta\in\erre} 
\sum_{k=1}^N (p_k-\theta)^2 \, \rho(\varepsilon^{-1}(x-Y_k))
\]
(see, e.g.,~\cite[p.~34]{Tsybakov}). Here ``local'' simply refers to the
weighting through $\rho$ that, as before, assigns more weights to
observations $Y_k$ closer to $x$.

\medskip

The explicit form of the Nadaraya-Watson estimator of $\pi_p$ we shall use
is
\[
\hat{\pi}_\varepsilon(K,T) = 
\frac{\ds \sum_{k=1}^N p_k \rho_{\varepsilon_1}(K-K_k) \rho_{\varepsilon_2}(T-T_k)}%
     {\ds \sum_{k=1}^N \rho_{\varepsilon_1}(K-K_k) \rho_{\varepsilon_2}(T-T_k)},
\]
where $\rho$ is the density of the standard Gaussian measure on $\erre$.

\medskip

As is well known (see, e.g.,~\cite{Tsybakov}), the choice of the
smoothing parameter $\varepsilon$ is crucial in the implementation of
non-parametric regression techniques: as $\varepsilon \to 0$ the
estimator is ``undersmoothing'', and as $\varepsilon \to \infty$ it is
``oversmoothing''.\footnote{Roughly speaking, as $\varepsilon \to 0$
  the estimator $\hat{f}_\varepsilon(x)$ reproduces the data, i.e. it
  is equal to $p_k$ for $x=Y_k$ and zero elsewhere, and it converges
  to a constant equal to the average of the $p_k$ as $\varepsilon \to
  \infty$.} We are going to select $\varepsilon$ by (leave-one-out)
cross-validation (cf.~\cite[{\S}1.4, p.~27-ff.]{Tsybakov}) as
follows: let $\hat{f}_{\varepsilon,-j}$ be the Nadaraya-Watson
estimator of $f$ based on the sample $(Y_k,p_k)_{k \in
  I\setminus\{j\}}$, $I:=\{1,\ldots,N\}$, and
\[
CV(\varepsilon) := 
\sum_{j \in I} \big\lvert f_j - \hat{f}_{\varepsilon,-j}(X_j) \big\rvert^2.
\]
Then we set
\[
\varepsilon_{CV} := \arg\min_{\varepsilon>0} CV(\varepsilon).
\]
In the NW estimator of the pricing functional we replace
$CV(\varepsilon)$ above with
\[
CV(\varepsilon) := 
\sum_{j \in I} \big\lvert 1 - \hat{f}_{\varepsilon,-j}(X_j)/f_j \big\rvert^2,
\]
i.e. the smoothing parameter is chosen as to minimize the relative
error rather than the absolute error.

Since the cross-validation procedure is computationally very
intensive, hence rather slow, we shall also consider a much simpler
choice of the smoothing parameters as a term of comparison. In
particular, following~\cite[\S{3.4.2}]{Silverman}, we shall use
\[
\varepsilon_j := 0.9 \, \min\bigl(D_j,Q_j/1.34\bigr) \, N^{-1/5}, \qquad j=1,2,
\]
where $D_1$ and $Q_1$ are the (sample) standard deviation and the
inter-quartile range of $K_k$, $k=1,\ldots,N$, respectively, and
$D_2$, $Q_2$ are defined in the same way with $T_k-t$ in place of $K_k$.

\begin{rmk}
  (a) In some applications (e.g. to estimate a function together with
  its derivatives) it is useful to allow $\rho$ to take negative
  values. Then $\hat{\pi}_\varepsilon$ is not guaranteed to be
  positive. The usual convention is then simply to redefine
  $\hat{\pi}_\varepsilon$ as its positive part. (b) If $\rho$ is
  supported on the whole space, the Nadaraya-Watson estimator is
  defined also for points that do not lie in the convex hull $C$ of
  $(Y_k)$. However, since this estimator is nothing else than a local
  average, estimates produced for points that lie outside $C$ should
  be taken with extreme care, if not plainly discarded.
\end{rmk}

\subsection{Implied volatility estimators}
\label{ssec:iv}
Let $\mathsf{BS}(S_0,r,q,K,x,\sigma)$ denote the Black-Scholes price of
a put option with time to maturity $x$ and strike $K$, written on an
underlying with current price $S_0$, constant dividend rate $q$ and
volatility $\sigma$, where the risk-free rate $r$ is also constant. As
is well known, $\sigma \mapsto \mathsf{BS}$ is strictly monotone,
hence, for any $S_0,x>0$, $K,r,q \geq 0$ and $0<p<K$, there exists a
unique $\sigma_0>0$, called \emph{implied volatility}, such that
$p=\mathsf{BS}(S_0,r,q,x,K,\sigma_0)$. We shall call by the same name
also the function $(p,S_0,r,q,x,K) \mapsto \sigma_0$ that is uniquely
defined by the procedure just described.

If $(p_k,K_k,x_k)$, $k=1,\ldots,N$, is a sample of observed option
prices and corresponding strike prices and times to maturity in a
fixed day (so that $S_0$ is also fixed), then for each $k$ there exists
a unique positive number $\hat{\sigma}_k$ such that
\[
p_k = \mathsf{BS}(S_0,r,q,K_k,x_k,\hat{\sigma}_k),
\]
hence $(\hat{\sigma}_k)_k$ can be interpreted as (estimates of the)
values of the implied volatility function on the set of points
$(Y_k)_k$, i.e., for some function $\sigma:\erre^2 \to \erre_+$,
$\hat{\sigma}_k = \sigma(Y_k)$ for all $k=1,\ldots,N$. This
yields
\[
p_k = \mathsf{BS}\bigl(K_k,x_k,\hat{\sigma}(K_k,x_k)\bigr) 
\qquad \forall k=1,\ldots,N,
\]
which immediately suggest another procedure to estimate the pricing
functional $\pi_p$: let $y_0 \in C$, where $C \subset \erre^2$ denotes
the convex hull of $(K_k,x_k)_{1\leq k\leq N}$, and define
$\hat{\sigma}(y_0)$ by linear interpolation of $\sigma$ (in the sense
of {\S}\ref{ssec:li}), hence set
\[
\hat{\pi}_p(y_0):=\mathsf{BS}\bigl(y_0,\hat{\sigma}(y_0)\bigr).
\]
In the empirical study below we shall employ a normalized estimator of
the implied volatility function $\sigma:C \to \erre$, in complete
analogy to the construction of the normalized linear interpolator of
\S\ref{ssec:li}. Namely, given $y_0 = (K_0,x_0) \in C$, we define
$\hat{\sigma}(y_0)$ by linear interpolation at $(K_0/S_0,x_0)$ of the
function whose value at $(K_k/S_0,x_k)$ is $\hat{\sigma}_k$,
$k=1,\ldots,N$.

Alternatively, one could use, in place of the linear interpolator
$\hat{\sigma}$, a Nadaraya-Watson estimator
$\hat{\sigma}_\varepsilon$, obtaining another estimator of the pricing
functional. Of course all considerations of the previous subsection
regarding the choice of the smoothing parameter, as well as the lack
of plausibility for estimates with $y_0 \not\in C$, apply also in this
case.

Note that, in contrast to the linear interpolator and the
Nadaraya-Watson estimator, estimating $\sigma_k$ requires estimates of
$r$ and $q$. While historical data on the risk-free interest rate
are readily available, we use as a proxy for $q$ the implied estimator
discussed at the end of Section \ref{sec:prelim} and, in more detail,
in Section \ref{sec:data} below. Moreover, the Black-Scholes structure
allows to obtain an explicit expression for the sensitivity of the
implied volatility with respect to the parameter $q$. Such result,
which could have some interest in its own right, can be found in the
Appendix (the derivation therein deals with call options, but the
corresponding result for put options follows easily).

\subsection{A parametric estimator of the Variance-Gamma class}
\label{ssec:VG}
A measure $\mu$ on the Borel $\sigma$-algebra of $\erre_+$ is called
Gamma measure with parameters $c>0$ and $\alpha>0$ if
\[
\mu(B) = \frac{\alpha^c}{\Gamma(c)} \int_B x^{c-1} e^{-\alpha x}\,dx
\]
for any Borel set $B$. A random variable $G$ is said to be
Gamma-distributed with parameters $c$ and $\alpha$ if its law is a
Gamma measure with the same parameters. Elementary calculus (and the
definition of the Gamma function) shows that, for any $w<\alpha$,
\begin{equation}
  \label{eq:mgf}
  \int_{\erre_+} e^{wx} \mu(dx) = (1 - w/\alpha)^{-c},  
\end{equation}
and
\[
\int_{\erre_+} e^{iux} \mu(dx) = (1 - iu/\alpha)^{-c},
\]
whence it immediately follows that Gamma laws are infinitely
divisible. Let $\mu_1$ be a Gamma measure with $c=1$. Then there
exists a positive increasing L\'evy process starting from zero (i.e.,
a subordinator) $\Gamma$ such that the law of $\Gamma_1$ is $\mu_1$,
and
\[
\E e^{iu\Gamma_t} = (1 - iu/\alpha)^{-t},
\]
hence the law of $\Gamma_t$ is Gamma with parameters $t$ and
$\alpha$. We refer to, e.g.,~\cite{Sato} for details.

Let $W$ a standard Wiener process independent of $\Gamma$, and
consider the process $X$ defined by
\[
X_t = \theta \Gamma_t + \sigma W_{\Gamma_t},
\]
where $\theta$ and $\sigma>0$ are constants. Since
$t \mapsto \theta t+ W_t$ is a L\'evy process and $X$ is obtained by
subordination of the former process, $X$ is itself a L\'evy process,
which goes under the name of (asymmetric) Variance-Gamma process and
was introduced in \cite{MCC}.

In order to construct a pricing functional, we are going to assume
that
\[
\theta + \frac{\sigma^2}{2} < \alpha.
\]
This condition guarantees that there exists a constant $\eta$ such
that the process $\exp(X_t + \eta t)$ is a
$\Q$-martingale.  In fact, since $W_{\Gamma_t}$ is equal in law to
$\Gamma_t^{1/2}W_1$, one has, recalling the expression for the
moment-generating function of Gaussian laws,
\begin{align*}
\E_\Q \exp\bigl(X_t\bigr) &= \E_\Q \exp\bigl(
\theta \Gamma_t + \sigma \Gamma_t^{1/2}W_1 \bigr)\\
&= \E_\Q \E_\Q\bigl[\exp\bigl(
\theta \Gamma_t + \sigma W_{\Gamma_t} \bigr) \big\vert \Gamma_t \bigr]\\
&= \E_\Q \exp\bigl( \theta\Gamma_t + \sigma^2/2 \Gamma_t \bigr)
= \E_\Q \exp\bigl( (\theta+ \sigma^2/2) \Gamma_t \bigr).
\end{align*}
Therefore, if $\theta+\sigma^2/2 < \alpha$, \eqref{eq:mgf} implies
\[
\E_\Q \exp\bigl(X_t\bigr) = \left(
1 - \frac{\theta+\sigma^2/2}{\alpha} \right)^{-t},
\]
hence $\E_\Q \exp(X_t+\eta t) = 1$ choosing
\[
\eta = \log \left( 1 - \frac{\theta+\sigma^2/2}{\alpha} \right).
\]
Since $X$ is a L\'evy process, it is now easy to conclude that the
process $\exp(X_t+\eta t)$ is a $\Q$-martingale.

We postulate that the risk-free interest rate $r$ and the dividend rate
$q$ are constant and that the price process $S$ can be written as
\[
S_t = S_0 \exp\bigl( (r-q)t + X_t + \eta t \bigr),
\]
so that the market satisfies the no-arbitrage condition. 
\begin{rmk}
  The hypothesis on $S$ just made is \emph{not} equivalent to assuming
  that
  \[
  S_t = S_0 + \int_0^t (r-q)S_s\,ds + \int_0^t S_{s-}\,dX_s.
  \]
  In fact, setting $\tilde{S}_t := e^{-(r-q)t}S_t$, $t \geq 0$, the
  integration-by-parts formula yields
  \[
  \tilde{S}_t = S_0 + \int_0^t \tilde{S}_{s-}\,dX_s,
  \]
  hence $\tilde{S}/S_0$ is given by the Dol\'eans stochastic
  exponential of $X$ (see, e.g.,~\cite[Theorem~26.8]{kall}), which in
  this case, recalling that $X$ is a pure-jump process with finite
  variation, reduces to
  \[
  \tilde{S}_t = S_0 \prod_{s\in]0,t]} \bigl(1+\Delta X_s\bigr),
  \]
  i.e.
  $S_t = S_0 e^{(r-q)t} \prod_{s\in]0,t]} \bigl(1+\Delta X_s\bigr)$.
  Since the Variance-Gamma process $X$ has paths of finite variation,
  the same conclusion can of course be obtained by elementary
  path-wise considerations, without any recourse to stochastic
  calculus.
\end{rmk}

The price at time zero of a European call option expiring at time $T$
is
\begin{align*}
  \pi_T &:= e^{-rT} \E_\Q\Bigl( S_0 \exp \bigl(
        (r-q)T + \eta T + X_T \bigr) - K \Bigr)^+\\
        &= e^{-rT} \E_\Q\Bigl( S_0 \exp \bigl(
        (r-q)T + \eta T + \theta\Gamma_T + \sigma \Gamma_T^{1/2}W_1 \bigr) 
        - K \Bigr)^+\\
        &= e^{-rT} \E_\Q\E_\Q\Bigr[\Bigl( S_0 \exp \bigl(
        (r-q)T + \eta T + \theta\Gamma_T + \sigma \Gamma_T^{1/2}W_1 \bigr) 
        - K \Bigr)^+ \Big\vert \Gamma_T \Bigr].
\end{align*}
Setting
\[
\tilde{q} \equiv \tilde{q}(\Gamma_T) := - (r-q+\eta)T 
- \bigl( \theta + \sigma^2/2 \bigr) \Gamma_T,
\qquad
\tilde{\sigma} \equiv \tilde{\sigma}(\Gamma_T) := \sigma \Gamma_T^{1/2},
\]
it is immediately seen that
\begin{align*}
  \pi_T &= e^{-rT} \E_\Q\E_\Q\Bigr[\Bigl( S_0 \exp \bigl(
        - \tilde{q} + \tilde{\sigma} W_1 - \tilde{\sigma}^2/2 \bigr) 
        - K \Bigr)^+ \Big\vert \Gamma_T \Bigr]\\
        &= e^{-rT} \E_\Q \mathsf{BS}\bigl(S_0,K,0,\tilde{q},\tilde{\sigma},1
                        \bigr)\\
        &= e^{-rT} \E_\Q \Bigl( S_0 \exp\bigl( -\tilde{q} \bigr)
           \Phi(\tilde{d}_+) - K \Phi(\tilde{d}_-) \Bigr)\\
        &= e^{-rT} \E_\Q \Bigl( S_0 \exp\bigl( 
                  (r - q + \eta)T + (\theta+\sigma^2/2) \Gamma_T \bigr)
           \Phi(\tilde{d}_+) - K \Phi(\tilde{d}_-) \Bigr)\\
        &= S_0 e^{(-q+\eta)T} \E_\Q\Bigl( 
           e^{(\theta+\sigma^2/2) \Gamma_T} \Phi(\tilde{d}_+) \Bigr)
           -K e^{-rT} \E_\Q \Phi(\tilde{d}_-),
\end{align*}
where
\begin{align*}
\tilde{d}_+ &= \frac{\log S_0/K + (r-q+\eta)T + (\theta+\sigma^2)\Gamma_T}%
                   {\sigma\Gamma_T^{1/2}},\\
\tilde{d}_- &= \frac{\log S_0/K + (r-q+\eta)T + \theta \Gamma_T}%
                   {\sigma\Gamma_T^{1/2}}.
\end{align*}
A completely similar argument shows that the price of a put option
with the same features can be written as
\[
K e^{-rT} \E_\Q \Phi(-\tilde{d}_-) - S_0 e^{(-q+\eta)T}
\E_\Q\Bigl( e^{(\theta+\sigma^2/2) \Gamma_T} \Phi(-\tilde{d}_+) \Bigr).
\]
The problem of option prices is thus reduced to the evaluation of
integrals against a Gamma measure. This can either be accomplished by
numerical integration (the Gamma measure has an explicit density and
the integrands, as functions of $\Gamma_T$, are ``almost'' explicit),
or, alternatively, by a simulation method. In particular, denoting a
sequence of independent copies of $\Gamma_T$ by $(G_k)$, the strong
law of large numbers yields
\[
\frac1n \sum_{k=1}^n F(G_k) \to \E_\Q F(\Gamma_T)
\]
$\Q$-almost surely for any (measurable) function $F:\erre \to \erre$
such that $\E_\Q\lvert F(\Gamma_T) \rvert < \infty$. Moreover, if
$F(\Gamma_T) \in L^2(\Q)$, the central limit theorem implies
\[
\frac{1}{\sqrt{n}} \sum_{k=1}^n \bigl( F(G_n) - \E_\Q F(\Gamma_T) \bigr)
\xrightarrow{d} N(0,\varsigma^2),
\]
where $\varsigma := \operatorname{Var} F(\Gamma_T)$. Writing
$\pi_T = F(\Gamma_T)$, for a properly chosen $F$, it is immediately
seen that $\pi_T \in L^1(\Q)$ thanks to the hypothesis on the
parameters $(\theta,\sigma,\alpha)$, and $\pi_T \in L^2(\Q)$ if
$2\theta+\sigma^2 < \alpha$. Since typically $\theta<0$ (negative
skew) and $\sigma$ rarely exceeds $1/2$, the condition $\alpha>1/4$ is
not restrictive (the estimated values of $\alpha$ in our data set are
always larger than $1$). For several representative choices of
parameters, an average over $n=10,000$ (pseudo)random variates produces
estimates of $\pi_T$ that are in very good agreement with those
obtained by numerical integration.

\medskip

Once a pricing formula for European options is available, calibration
of the parameters is rather straightforward. Namely, denoting by
$\pi_{VG}(K,T-t;\theta,\sigma,\alpha)$ the (theoretical) price of a
put option in the above VG model (with $t$ fixed), assuming that
$(K_k,T_k-t)$ and $p_k$, $k=1,\ldots,N$, are observed parameters
corresponding and option prices in a fixed day, one sets
\[
(\hat{\theta},\hat{\sigma},\hat{\alpha}) := 
\arg\min_{(\theta,\sigma,\alpha) \in D} \sum_{k=1}^N \left\vert 
\frac{\pi_{VG}(K_k,T_k-t;\theta,\sigma,\alpha)-p_k}{p_k} \right\vert^2,
\]
where
$D := \bigl\{ (\theta,\sigma,\alpha) \in \erre \times (0,\infty)^2 :
\theta + \sigma^2/2 < \alpha \bigr\}$.
This procedure, possibly with the sum of absolute errors instead of
relative errors, is widely used by practitioners as well as in
academic publications (cf.~e.g.~\cite{CarrWu:03} and \cite{CGMY},
respectively).

\begin{rmk}
  It should be pointed out, however, that the map
  $(\theta,\sigma,\alpha) \mapsto \pi_{VG}$ is \emph{not} injective,
  hence the calibration procedure just outlined is not well posed,
  i.e. the (daily) estimates $\hat{\theta}$, $\hat{\sigma}$,
  $\hat{\alpha}$ are not unique. In practice they will depend on the
  initialization of the minimization algorithm (we have chosen
  $\theta_0=0$, $\sigma_0=0.3$ and $\alpha_0 = 2$, respectively).
\end{rmk}

%-------------------------------------------------------------------

\section{Data}
\label{sec:data}
We use S\&P500 index option data\footnote{The raw data are obtained
  from \emph{Historical Option Data}, see
  \texttt{www.historicaloptiondata.com}.} for the period January 3,
2012 to December 31, 2012. The sample contains $77\,408$ observations of
European call and put options. Prices are averages of bid and ask
prices.  Data points with time-to-maturity less than one day or volume
less than $100$ are eliminated. This reduces the size of the sample to
$75\,022$: 61\% and 39\% of the options are of put and call type,
respectively.

During 2012 the annualized mean and standard deviation of daily
returns of the S\&P500 index were equal to $11.09\%$ and $12.64\%$,
respectively. During the same period the 1-year T-bill rate was very
close to zero, with minimal variations: in particular, its mean was
equal to $0.16\%$, with a standard deviation equal to $0.023\%$.

As is well known, index options on the S\&P500 are very actively
traded: the average daily volume is $618\,490$ contracts, with
maturities ranging from $1$ day to almost $3$ years. Descriptive
statistics of the data are collected in Table
\ref{table1}.

\begin{table}[tbp]
\caption{Summary statistics for S\&P500 index options data}
\begin{center}
\vspace{2mm}
\parbox{\textwidth}{\footnotesize This table collects some simple
  statistics for prices of European call and put options on the
  S\&P500 index. The sample period is January 3, 2012 to December 31,
  2012. Implied volatilities are annualized, time to maturity is
  expressed in days, strike and futures prices are expressed in index
  points.  \vspace{4mm}}
\footnotesize{{\begin{center} 
 \begin{tabular*}{\textwidth}{@{\extracolsep{\fill}}lccccccccc} 
 \toprule 
 \multicolumn{4}{c}{} & \multicolumn{5}{c}{Percentiles} & \multicolumn{1}{c}{}\\ 
 \cline{5-9}\\ 
 \multicolumn{4}{c}{}\\ 
Variable & Mean & Std& Min &$5\%$& $10\%$& $50\%$ & $90\%$ & $95\%$& Max  \\ 
 \midrule 
Call price             & 34.3& 98.8& 0.0& 0.1& 0.2& 9.2& 75.2& 115.5& 1270.0\\ 
Put price              & 21.3& 46.5& 0.0& 0.1& 0.1& 5.9& 58.2& 93.7& 1197.0\\ 
Implied vol.       & 0.2& 0.1& 0.0& 0.1& 0.1& 0.2& 0.4& 0.4& 2.6\\ 
Implied ATM vol.    & 0.2& 0.1& 0.0& 0.1& 0.1& 0.2& 0.4& 0.4& 2.0\\ 
Time to maturity & 96.7& 157.0& 1.0& 2.0& 4.0& 38.0& 269.0& 404.0& 1088.0\\ 
Strike price       & 1301.0& 208.4& 100.0& 950.0& 1075.0& 1345.0& 1480.0& 1525.0& 3000.0\\ 
Futures price       & 1374.4& 48.4& 1207.2& 1289.5& 1309.1& 1377.1& 1435.9& 1450.7& 1466.8\\ 
     \hline 
\bottomrule 
\end{tabular*} 
\end{center}}}
\end{center}
\label{table1}
\end{table}

It is commonly accepted that prices of in-the-money (ITM) options,
because of their small trading volume, are not reliable, and that, for
this reason, they should be replaced by prices computed via put-call
parity whenever possible: the ``new'' prices are determined by prices
of out-of-the-money (OTM) options, that are generally traded in larger
volumes, and are hence considered to be accurately priced (cf.,
e.g.,~\cite[p.~517-ff.]{AL} and \cite{CarrWu:03}). We need therefore
to check whether our data set is affected by such phenomenon. In other
words, we need to check whether the recorded prices for ITM options
satisfy the put-call parity relation with the corresponding OTM
options. We are going to show that prices of ITM options in our data
set can be considered perfectly reliable, and hence that no correction
is needed. It is natural to argue that discarding prices of options
whose trading volume is lower than 100 already eliminates possibly
unreliable quotes.

Basic summary statistics on the volume of the options in our database
according to their moneyness (see below for the precise definition we
adopt) are collected in Table \ref{table1bis}. Note that the total
number of traded ITM and OTM options are rather close.

\begin{table}[tbp]
\caption{Statistics on the daily volume for S\&P500 index options}
\begin{center}
\vspace{2mm}
\parbox{\textwidth}{\footnotesize This table collects basic
  descriptive statistical figures on the daily trading volume of
  European call and put options on the S\&P500 index. The sample
  period is January 3, 2012 to December 31, 2012. Moneyness is defined
  in terms of the spot price falling within a $10\%$ interval centered
  around the strike price.  \vspace{4mm}}
\footnotesize{{\begin{center} 
 \begin{tabular*}{\textwidth}{@{\extracolsep{\fill}}lccccccccc} 
 \toprule 
 \multicolumn{4}{c}{} & \multicolumn{5}{c}{Percentiles} & \multicolumn{1}{c}{}\\ 
 \cline{5-9}\\ 
 \multicolumn{4}{c}{}\\ 
Moneyness & Mean & Std& Min &$5\%$& $10\%$& $50\%$ & $90\%$ & $95\%$& Max  \\ 
 \midrule 
At-the-money & 2696& 4793& 100& 110& 150& 825& 7824& 12475& 62334\\ 
In-the-money & 1489& 3972& 100& 100& 113& 490& 3000& 6000& 65412\\ 
Out-of-the-money& 1688& 3191& 100& 108& 133& 640& 4037& 6381& 101718\\ 
     \hline 
\bottomrule 
\end{tabular*} 
\end{center}}}
\end{center}
\label{table1bis}
\end{table}

Let us describe in detail the procedure to obtain ``better'' prices of
ITM options via prices of corresponding OTM options: let $t<T$, $S_t$
and $C_t$ be given, where $S_t$ and $C_t$ denote the observed prices
at time $t$ of the index and of an ITM call option\footnote{The
  reasoning obviously holds, \emph{mutatis mutandis}, also for an ITM
  put option.} with maturity $T$ and strike price $K$,
respectively. Since the put option with \emph{the same} maturity and
strike is necessarily OTM, if Hypothesis (H) is satisfied (or simply
assumed to hold), the put-call parity identity \eqref{eq:parity}
provides an ``alternative'' price for the ITM call option, provided
that
\begin{itemize}
\item[(a)] a put option with \emph{the same} maturity and strike is
  traded;
\item[(b)] estimates of $\bar{r}(t,T)$ and $\bar{q}(t,T)$ are
  available.
\end{itemize}
While good estimates of $\bar{r}(t,T)$ are easily available,
estimating $\bar{q}(t,T)$ is in general not straightforward. In
particular, since we are going to use the estimator defined in
\eqref{eq:qest}, which relies on put-call parity, one needs to avoid
circular reasoning. This is achieved by using pairs of at-the-money
(ATM) put and call options with \emph{the same} maturity and strike to
estimate $\bar{q}(t,T)$. The latter estimates are then used in the
put-call parity formula for ITM/OTM options, so that no circularity is
involved, as, for any given day, the sets of ATM, ITM and OTM options
are disjoint. ATM options are in general very liquid, so that their
observed prices can be considered accurate. This implies that the
corresponding estimates of $\bar{q}$ can also be considered
accurate. On the other hand, it may happen that for an ITM option with
maturity $T$ there is no available estimate of $\bar{q}(t,T)$, simply
because no couple of ATM options with that maturity is traded. In this
case we use linear interpolation, if possible, and nearest-neighbor
extrapolation otherwise.

To implement the procedure just outlined, it is clearly necessary to
define a measure of moneyness, so that options can be (uniquely)
classified as at-the-money, in-the-money, or out-of-the-money. The
simplest definition of (logarithmic spot) moneyness at time $t$ for a
European call or put option with maturity $T$ and strike $K$ is
$\log K/S_t$.  Closely related is the logarithmic forward simple
moneyness, defined as $\log K/f^{S_T}_t$, that is clearly better suited
especially for options with longer maturities. However, recalling that
the forward price $f^{S_T}_t$ depends on $\bar{q}(t,T)$, the risk of
falling into a circular reasoning appears again. To avoid this problem
we simply use as moneyness
\[
M(t,T,K) = \log (K/f_t), \qquad f_t := S_t e^{(r_t-q_t)(T-t)},
\]
where $r_t$ is the spot rate at time $t$ and $q_t$ is the (historical
estimate\footnote{Such estimates of $q_t$, in the case of the S\&P500,
  are readily available.} of the) dividend rate at time $t$.  Then we
say that an option is \emph{at-the-money} (ATM) if its moneyness lies
in the interval $[\log 0.95,\log 1.05]$, i.e. if its forward price at
time $t$ falls within a $10\%$ interval centered around its strike
price.  It is perhaps worth noting that, according to
\cite{CarrWu:03}, the industry-standard definition of moneyness is the
so-called standardized forward moneyness, defined as
\[
\frac{\log(K/f^{S_T}_t)}{\sigma_i \sqrt{T-t}},
\]
where $\sigma_i$ is the (Black-Scholes) implied volatility of the
option.\footnote{In this formula one could replace $f^{S_T}_t$ with
  $f_t$ as defined above, however, the further problem of having to
  estimate the implied volatility appear.}  At any rate, the empirical
results discussed below are essentially insensitive to the definition
of moneyness. 

\medskip

After having spelled out in detail all steps in the implementation of
put-call parity to deduce prices of ITM options from those of
corresponding OTM options, we can now substantiate our claim that ITM
prices can already be considered reliable and no correction is
needed. In fact, denoting by $p$ the market price, by $\tilde{p}$ the
``parity'' price, and defining the relative error by
$\lvert(\tilde{p}-p)/p\rvert$, the mean relative error is $0.3\%$ with
a standard deviation of $0.4\%$. Moreover, the relative error is less
than $1\%$ in $97.5\%$ of cases, with a maximum value equal to
$4\%$. It seems therefore perfectly fine to accept the quoted prices
of ITM option as reliable.

\begin{rmk}
  The procedure described above to ``correct'' ITM option prices is
  used in \cite{AL} (up to details regarding the definition of
  moneyness, which is not explicitly stated therein). The options in
  their data set, being about 20 years older than ours, had much
  smaller trading volumes, at least on average. This might explain why
  they did need to replace prices of ITM options in their database via
  the above correction procedure.
\end{rmk}

\subsection{Option prices across different dates}
\label{ssec:markov}
We argued in Section \ref{sec:estim} that it does not seem reasonable
to impose (Markov) time-homogeneity assumptions on the data-generating
process. To substantiate this claim, we identified two different dates
in our data set (February 2, 2012 and June 21, 2012) when the prices
of the S\&P500 were practically identical ($\$1325.54$ and
$\$1325.51$, respectively), and we looked for pairs of options, traded
on both dates, having the \emph{same} strike price and time to
maturity. Since the risk-free rate as well as the dividend rate were
very close to constant all over the year, a time-homogeneity
assumption could be taken into consideration if the above-mentioned
pairs of options with identical characteristics had prices very close
to each other. Unfortunately this is clearly not the case: in Figure
\ref{fig1} prices of pairs of options, all having time to maturity
equal to $8$ days, are plotted as function of their common strike
price. From a practical viewpoint, the marked price difference could
be explained, at least in part, by different levels of volatility on
the two trading days. In fact, the CBOE Volatility Index for the two
dates was equal to 17.98 and 20.08, respectively.
\begin{figure}[tbp]
\caption{Non-stationarity of the price process}
\label{fig1}\vspace{2mm}
\parbox{\textwidth}{\footnotesize This figure shows the difference in
  prices of options with same time to maturity and strike price,
  comparing two days (February 2 and June 21, 2012) when the prices of
  the underlying coincide.}
\newline
\vspace{2mm}
\par
\begin{center}
\includegraphics[width=0.85\textwidth]{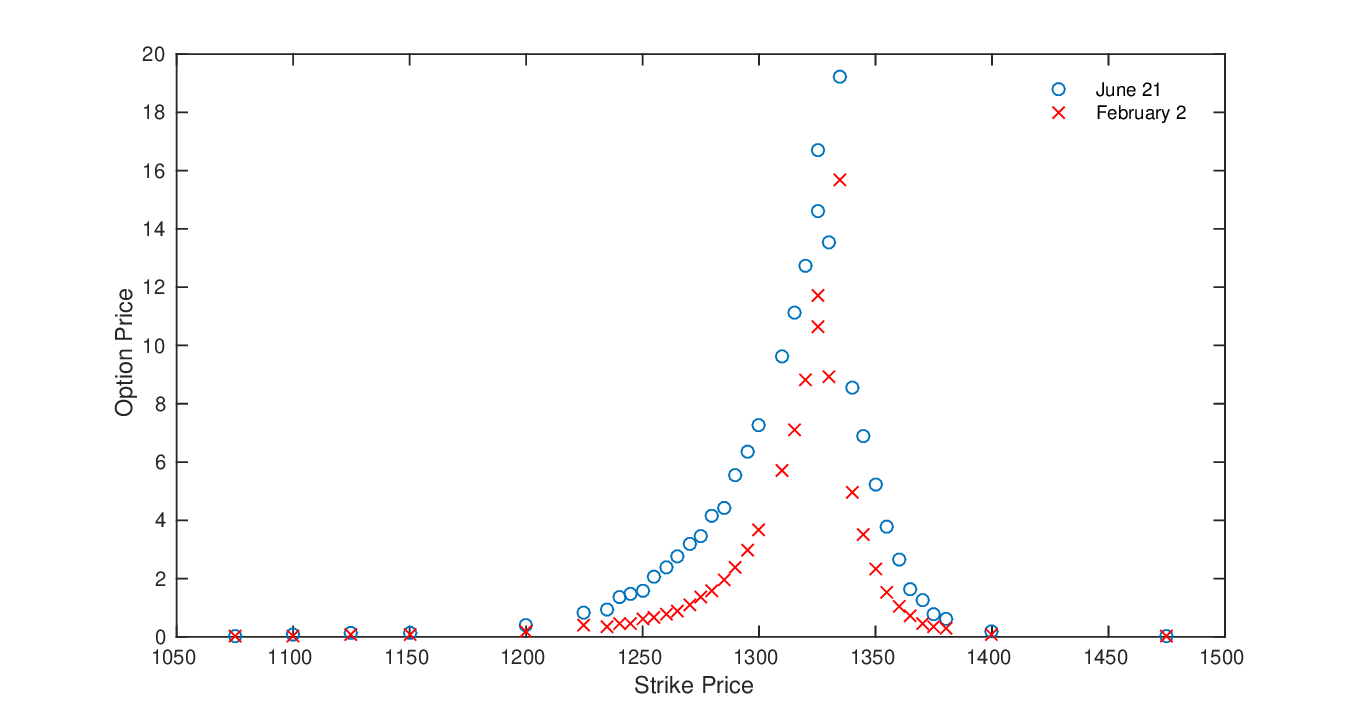}
\end{center}
\end{figure}

%-------------------------------------------------------------------

\section{Empirical analysis}
\label{sec:ea}
In this section we compare the empirical performance of the estimators
introduced above. The Variance-Gamma estimator, being a simple yet
non-trivial parametric approach, can be considered as a
benchmark.

The analysis is carried out on the full data set as well as on a
trimmed data set, where options with price lower than 1/8 or implied
volatility higher than 0.7 are discarded a priori (the latter
preliminary screening follows the approach of \cite{AL}). Moreover, we
consider both put and call options, that is, in one case we discard a
priori all call prices, and in another case all put prices. We have
hence four different data sets on which the (daily) performance of the
estimators is tested.

We are going to perform out-of-sample tests as follows: let
$(p_k)_{k\in J_t}$ be observed put option prices in a fixed trading
day $t$, with $J_t=J_t^0 \cup J_t^1$,
$J_t^0 \cap J_t^1 = \varnothing$. The index set $J_t^0$ is used to
construct estimators of option prices and is randomly chosen such that
its size is 90\% (up to rounding to the next integer) of the size of
$J_t$. Since $t$ is fixed, under the hypotheses of Section
\ref{sec:estim} we can write, without loss of generality,
$p_k = \pi_p(K_k,T_k-t)$ for all $k \in J_t$.  Denoting any of the
estimators of Section \ref{sec:estim} by $\hat{\pi}$, constructed on
the basis of $(p_k,K_k,T_k-t)_{k \in J_t^0}$, we obtain estimates
$\hat{p}_k=\hat{\pi}(K_k,T_k-t)$ for all $k \in J_t^1$. These estimates
are well defined, or meaningful, only for those $k \in J_t^1$ such that
$(K_k,T_k-t)$ belongs to the convex hull of $(K_k,T_k-t)_{k\in J_t^0}$
whenever a linear estimator or a Nadaraya-Watson estimator is used.
The (relative) error for any one of the estimators is then simply
defined as $e_k = \lvert 1- \hat{p}_k/p_k \rvert$ for all
$k \in J_t^1$ such that $\hat{p}_k$ is well defined. This procedure is
repeated for all trading days of the year. We the compute various
averages and related statistics on the obtained sample of relative
errors $(e_k)_k$ for the whole year.

\begin{rmk}
  Discarding all call prices to construct estimators of the pricing
  functional for put prices (and vice versa) implies that a large
  amount of potentially useful information is being discarded. In
  fact, in view of put-call parity, prices of call options can be
  translated into prices of put options with the same maturity and
  strike price, exactly as in the procedure described at length in
  Section \ref{sec:data}. However, given the usual characteristics of
  call and put options on the S\&P500, using the additional
  information does not improve essentially the accuracy of the
  estimators. In fact, as one can see in Figure \ref{fig2}, the set of
  points $(K,x)$ for call and put options, where $K$ and $x$ stand for
  strike price and time to maturity, respectively, are almost
  disjoint. It is clear, on the other hand, that taking into account
  the information contained in the whole set of option prices could be
  very useful to price (deep) ITM options.
\end{rmk}

\begin{figure}[tbp]
\caption{Time to maturity and strike prices: call vs. put}
\label{fig2}\vspace{2mm}
\parbox{\textwidth}{\footnotesize This figure shows time to maturity
  and strike price of all options traded in one day (June 21, 2012).}
\newline
\vspace{2mm}
\par
\begin{center}
\includegraphics[width=0.85\textwidth]{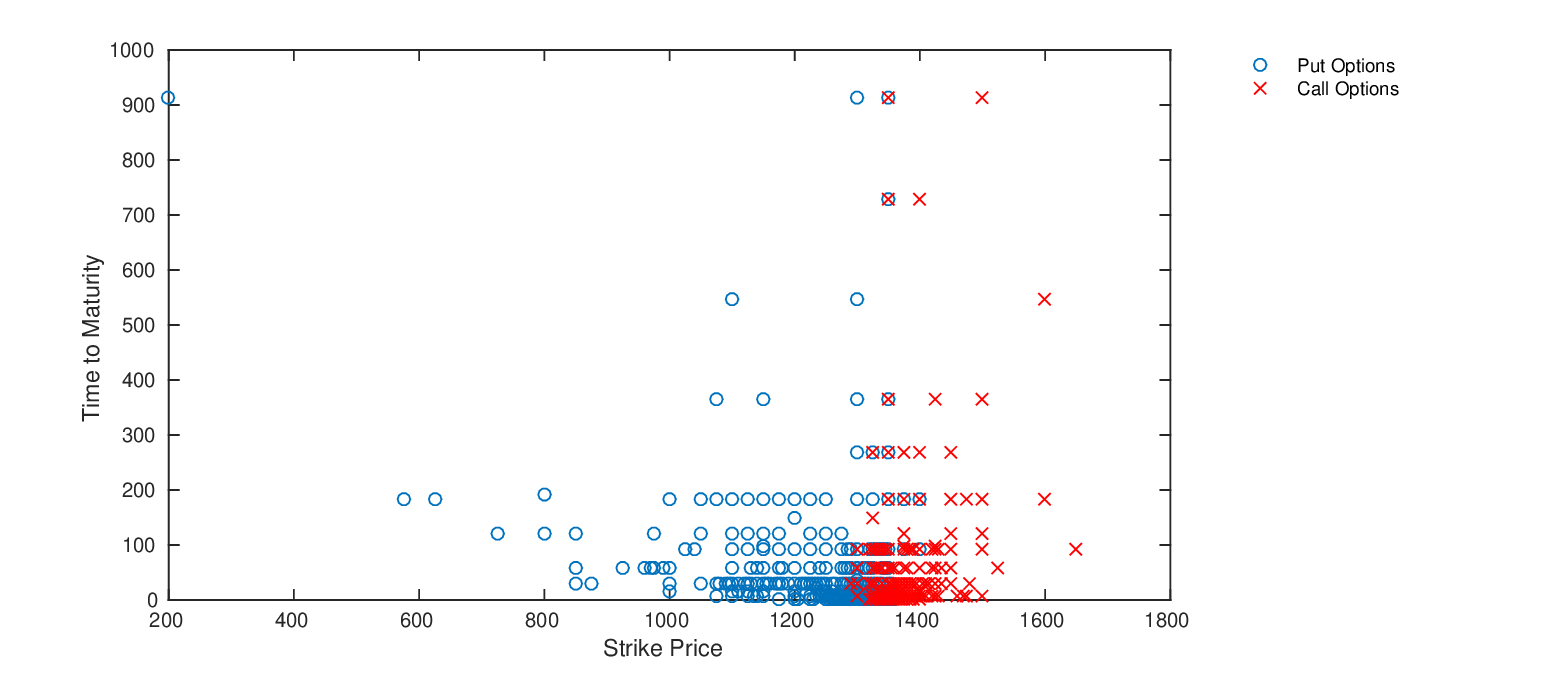}
\end{center}
\end{figure}

Before presenting the empirical results, let us briefly introduce the
labels used below and in the corresponding tables:
\begin{description}\label{etichette}
\item[LI] Normalized linear interpolator of the pricing function
  defined in \S\ref{ssec:li}.
\item[BS] Black-Scholes price with volatility obtained by
  normalized linear interpolation of the volatility surface, as
  defined in \S\ref{ssec:iv}.
\item[NW] Nadaraya-Watson estimator of the pricing
  function with smoothing parameters chosen by a quantile-based
  criterion, as defined in \S\ref{ssec:NW}.
\item[NW\raisebox{-.4ex}{\scriptsize CV}] As NW, but with smoothing
  parameter chosen by cross-validation.
\item[BS-NW] Black-Scholes price with volatility obtained by
  Nadaraya-Watson estimation.
\item[BS-NW\raisebox{-.4ex}{\scriptsize CV}] As BS-NW, but with
  smoothing parameter chosen by cross-validation.
\item[VG] The Variance-Gamma estimator as defined in \S\ref{ssec:VG}.
\end{description}

\subsection{Pricing within the hull}
Statistics on the empirical pricing error computed on the whole data
set, both for the case of put and of call options, are collected in
Table \ref{table2}. Since, as mentioned above, the linear interpolator
of the pricing functional (LI) as well as the the linear interpolator
of the implied volatility function (hence BS) are only defined on the
convex hull of the sets of couples $(K_k,T_k-t)_{k \in J_t^0}$, for
all $t$, the set of empirical errors $(e_k)_k$ is restricted to those
$k$ such that the linear interpolator is well defined. This
restriction implies that $5.1\%$ of the put options and $7.4\%$ of the
call options (whose index $k$ belongs to $J_t^1$ for some $t$) cannot
be priced by methods relying on linear interpolation. The behavior of
the estimators that do not suffer of this restriction (at least
formally) is discussed later.

It turns out that all non-parametric estimators, if evaluated in terms
of their mean $L_1$ relative error, are not satisfactory, with figures
ranging from around 9\% to 40\%. The parametric estimator based on the
Variance-Gamma process (VG) has even higher mean errors.\footnote{It
  may be interesting to note that the mean relative in-sample and
  out-of-sample errors of the VG estimator are very close. For
  example, the mean in-sample errors for put and call options on the
  trimmed dataset are equal to $39.6\%$ and $48.2\%$, respectively,
  while the corresponding mean out-of-sample errors are equal to
  $38.5\%$ and $49.6\%$. It is reasonable to speculate that the VG
  process is simply a poor fit to the data-generating
  process. However, one should also recall that the calibration
  procedure of the VG model is \emph{not} well posed.}

\begin{table}[tbp]
\caption{Pricing errors on the full dataset}
\begin{center}
\vspace{2mm}
\parbox{\textwidth}{\footnotesize The table displays descriptive
  statistics on the empirical pricing errors (in percentage points) of
  the various estimators, on the basis of the complete dataset. Some
  points of the empirical distribution of the pricing error are also
  reported.  The labels used for the different estimators are defined
  on page~\pageref{etichette}. The sample period is January 3, 2012 to
  December 31, 2012. The pricing error is computed on 4460 put options
  and on 2732 call options. \vspace{4mm}}
\footnotesize{{\begin{center} 
 \begin{tabular*}{\textwidth}{@{\extracolsep{\fill}}lccccccc} 
 \toprule 
 \multicolumn{8}{c}{\textsc{Put Options}}\\ 
 \multicolumn{8}{c}{}\\ 
Error & LI & BS& NW & BS-NW & NW\raisebox{-.4ex}{\tiny CV} & BS-NW\raisebox{-.4ex}{\tiny CV} & VG \\  
 \midrule 
Mean    & 11.9& 10.7& 289.1& 39.8& 29.5& 32.2& 46.9\\ 
St. Dev.& 49.0& 130.8& 985.7& 249.4& 113.2& 280.5& 28.0\\ 
Median  & 3.1& 1.1& 26.9& 16.0& 10.8& 6.4& 45.7\\ 
Min     & 0.0& 0.0& 0.0& 0.0& 0.0& 0.0& 0.0\\ 
Max     & 2027.0& 8079.8& 17792.1& 15551.7& 4649.3& 16892.0& 314.7\\ 
 \multicolumn{8}{c}{}\\ 
 \multicolumn{8}{c}{Empirical Distribution}\\ 
 \midrule 
$1\%$   & 23.7& 48.0& 3.1& 6.2& 8.1& 16.3& 0.9\\ 
$5\%$   & 61.6& 73.4& 13.4& 25.8& 30.2& 44.2& 4.6\\ 
$10\%$  & 75.9& 82.9& 24.5& 40.0& 47.4& 60.0& 9.3\\ 
$20\%$  & 87.5& 91.2& 41.6& 55.0& 69.8& 73.6& 18.9\\ 
$25\%$  & 90.0& 93.0& 48.0& 59.6& 75.7& 77.0& 23.9\\ 
$30\%$  & 92.2& 94.6& 52.6& 63.5& 80.2& 79.9& 29.3\\ 
$50\%$  & 96.0& 97.2& 65.3& 75.4& 89.3& 87.0& 56.6\\ 
 \multicolumn{8}{c}{}\\ 
 \multicolumn{8}{c}{}\\ 
 \toprule 
 \multicolumn{8}{c}{\textsc{Call Options}}\\ 
 \multicolumn{8}{c}{}\\ 
Error & LI & BS& NW & BS-NW & NW\raisebox{-.4ex}{\tiny CV} & BS-NW\raisebox{-.4ex}{\tiny CV} & VG \\  
 \midrule 
Mean    & 30.1& 8.7& 278.3& 26.6& 40.0& 14.7& 57.4\\ 
St. Dev.& 418.0& 36.6& 1013.1& 48.1& 180.0& 51.1& 30.3\\ 
Median  & 5.0& 1.3& 29.2& 7.9& 15.6& 4.2& 64.8\\ 
Min     & 0.0& 0.0& 0.0& 0.0& 0.0& 0.0& 0.0\\ 
Max     & 16075.0& 1125.9& 16053.6& 830.2& 4733.0& 1811.3& 276.1\\ 
 \multicolumn{8}{c}{}\\ 
 \multicolumn{8}{c}{Empirical Distribution}\\ 
 \midrule 
$1\%$   & 18.9& 45.8& 2.4& 10.0& 7.5& 19.4& 1.9\\ 
$5\%$   & 50.0& 73.6& 10.0& 38.4& 23.5& 55.3& 5.2\\ 
$10\%$  & 67.9& 83.4& 17.9& 54.7& 35.7& 71.4& 9.0\\ 
$20\%$  & 81.8& 91.1& 36.7& 68.1& 58.7& 84.2& 16.4\\ 
$25\%$  & 85.4& 93.3& 45.2& 72.2& 65.8& 87.0& 19.5\\ 
$30\%$  & 87.8& 94.9& 50.9& 74.7& 72.6& 88.9& 22.5\\ 
$50\%$  & 93.2& 97.4& 64.6& 82.7& 85.8& 94.2& 36.3\\ 
\bottomrule 
\end{tabular*} 
\end{center}
}}
\end{center}
\label{table2}
\end{table}

A more accurate evaluation of the performance of the estimators is
obtained by looking at the distribution of the empirical errors, some
points of which are also displayed in Table \ref{table2}. In
particular, one observes that the best performance, in terms of the
number of options whose estimated price is reasonably near the market
price, is achieved by the estimator based on Black-Scholes formula
with linearly interpolated volatility (BS). It is also evident that
the LI estimator has a considerably worse performance that the BS
estimator. Similarly, the Nadaraya-Watson estimator of the pricing
functional displays much weaker performance than the BS
estimator, both when the smoothing parameter is chosen via a simple
quantile-based criterion and when it is chosen via cross
validation. In spite of the major improvement of the NW-based methods
when cross validation is used with respect to their simpler
quantile-based counterparts, \emph{all} NW-based methods perform
considerably worse than both methods based on simple linear
interpolation, that is LI and BS. Taking into account that the choice
of the smoothing parameter by cross validation is computationally much
more intensive (and slower) than linear interpolation, these empirical
results suggest that, for the purpose of the pricing problem treated
here, keeping things simple is not only faster, but also significantly
more accurate. At this point we should remark that the empirical
results reported in this section are obviously influenced by the
(random) splitting of the index set $J_t$ in two disjoint sets $J_t^0$
and $J_t^1$ for each date $t$. However, while changing the
initialization of the random number generator (obviously) results in
different figures, the qualitative observations just made, as well as
the conclusions they imply, do not change.

\medskip

The whole analysis has also been carried out on a smaller data set,
obtained by eliminating those options whose implied volatility is
higher than $70$\% or price is lower than
\$$1/8$. This procedure is adopted in \cite{AL} and is repeated here
only for comparison purposes with the results summarized in Table
\ref{table2} (recall, however, that, as discussed above, the analysis
of \cite{AL} and ours are not directly comparable). The corresponding
results are collected in Table \ref{table3}. 

\begin{table}[tbp]
\caption{Pricing errors on the trimmed dataset}
\begin{center}
\vspace{2mm}
\parbox{\textwidth}{\footnotesize The table displays descriptive
  statistics on the empirical pricing errors (in percentage points) of
  the various estimators, on the basis of the trimmed dataset. Some
  points of the empirical distribution of the pricing error are also
  reported.  The labels used for the different estimators are defined
  on page~\pageref{etichette}. The sample period is January 3, 2012 to
  December 31, 2012. The pricing error is computed on 4021 put options
  and on 2496 call options. \vspace{4mm}}
\footnotesize{{\begin{center} 
 \begin{tabular*}{\textwidth}{@{\extracolsep{\fill}}lccccccc} 
 \toprule 
 \multicolumn{8}{c}{\textsc{Put Options}}\\ 
 \multicolumn{8}{c}{}\\ 
Error & LI & BS& NW & BS-NW & NW\raisebox{-.4ex}{\tiny CV} & BS-NW\raisebox{-.4ex}{\tiny CV} & VG \\  
 \midrule 
Mean       & 9.5& 4.8& 151.4& 20.8& 22.8& 17.9& 38.5\\ 
St. Dev.   & 35.5& 13.3& 425.5& 29.0& 62.6& 72.8& 24.6\\ 
Median     & 3.2& 1.0& 24.8& 10.2& 10.8& 5.5& 35.9\\ 
Min        & 0.0& 0.0& 0.0& 0.0& 0.0& 0.0& 0.0\\ 
Max        & 1129.3& 322.9& 6152.4& 527.5& 1462.2& 3358.0& 202.0\\ 
 \multicolumn{8}{c}{}\\ 
 \multicolumn{8}{c}{Empirical Distribution}\\ 
 \midrule 
$1\%$      & 21.7& 50.1& 3.1& 10.0& 8.6& 17.7& 1.2\\ 
$5\%$      & 62.0& 78.6& 13.2& 33.6& 31.1& 48.1& 5.8\\ 
$10\%$     & 79.6& 88.0& 25.0& 49.6& 48.0& 64.3& 12.0\\ 
$20\%$     & 91.3& 95.1& 43.0& 67.0& 71.1& 78.6& 25.2\\ 
$25\%$     & 93.6& 96.4& 50.2& 73.0& 78.3& 82.3& 31.9\\ 
$30\%$     & 94.6& 97.1& 55.7& 77.4& 83.0& 85.2& 40.0\\ 
$50\%$     & 97.6& 98.7& 68.9& 87.6& 91.6& 91.9& 72.5\\ 
 \multicolumn{8}{c}{}\\ 
 \multicolumn{8}{c}{}\\ 
 \toprule 
 \multicolumn{8}{c}{\textsc{Call Options}}\\ 
 \multicolumn{8}{c}{}\\ 
Error & LI & BS& NW & BS-NW & NW\raisebox{-.4ex}{\tiny CV} & BS-NW\raisebox{-.4ex}{\tiny CV} & VG \\  
 \midrule 
Mean       & 14.4& 5.2& 138.4& 14.3& 31.3& 9.3& 49.6\\ 
St. Dev.   & 45.2& 18.7& 375.1& 20.0& 77.0& 17.6& 27.6\\ 
Median     & 4.6& 1.1& 27.3& 6.8& 16.1& 3.8& 54.2\\ 
Min        & 0.0& 0.0& 0.0& 0.0& 0.0& 0.0& 0.0\\ 
Max        & 1233.6& 649.1& 4954.1& 217.9& 1952.5& 315.4& 227.0\\ 
 \multicolumn{8}{c}{}\\ 
 \multicolumn{8}{c}{Empirical Distribution}\\ 
 \midrule 
$1\%$      & 17.7& 47.2& 2.5& 9.7& 7.6& 19.2& 2.3\\ 
$5\%$      & 51.8& 79.2& 11.1& 41.5& 23.2& 57.1& 6.7\\ 
$10\%$     & 69.4& 88.1& 19.7& 60.8& 34.9& 75.1& 11.7\\ 
$20\%$     & 83.5& 94.6& 38.5& 79.1& 57.9& 88.5& 19.3\\ 
$25\%$     & 87.3& 96.0& 46.4& 82.6& 65.6& 91.2& 23.2\\ 
$30\%$     & 89.5& 96.4& 53.6& 86.1& 72.5& 93.2& 27.2\\ 
$50\%$     & 94.2& 98.4& 69.0& 93.7& 86.2& 97.0& 45.0\\ 
\bottomrule 
\end{tabular*} 
\end{center}
}}
\end{center}
\label{table3}
\end{table}

Even though, as is natural to expect, the mean pricing error improves
for all estimators, the qualitative picture emerged from the analysis
of the whole data set does not change. In particular, the only
NW-based method with a barely acceptable performance is the
Black-Scholes estimator coupled with estimation of the implied
volatility when the smoothing parameter is chosen by cross
validation. The important qualitative observation, though, is, as
before, that keeping things simple, especially using the BS estimator,
is both faster and more accurate in the vast majority of
cases. Another important qualitative conclusion is that pricing by
non-trivial fully parametric models, such as the Variance-Gamma model
of \S\ref{ssec:VG} , calibrated to observed option prices, does not
produce reliable estimates, even though the model allows for skewed
and (moderately) heavy-tailed distributions of returns.

%\begin{table}[tbp]
%\caption{Difference of pricing errors on the trimmed dataset}
%\begin{center}
%\vspace{2mm}
%\parbox{\textwidth}{\footnotesize The table displays statistical 
%  tests on the difference of pricing errors of
%  the various estimators, on the basis of the trimmed dataset. 
%The first panel shows p-values of a two-sample Kolmogorov-Smirnov test, while the second panel shows p-values of a sign test on zero median for the difference of estimators. 
%The sample period is January 3, 2012 to
%  December 31, 2012. The pricing error is computed on 4021 put options
%  and on 2496 call options. \vspace{4mm}}
%  \scriptsize {{\input{t3_test.tex}}}
%\end{center}
%\label{table3_test}
%\end{table}

\medskip

The empirical distribution of the (absolute) pricing errors clearly
indicates that the mean errors are heavily influenced by relatively
few large values. This observation holds both for the whole as well as
for the reduced dataset. For instance, the mean pricing error of the
LI method for put options is around $12$\% and $9.5$\%, respectively,
while more than $75$\% of the errors are less than $10$\% in both
cases. In particular, by direct inspection of the results obtained for
a few randomly chosen dates, one observes that high relative errors
mostly afflict options with low prices. It appears therefore
interesting to analyze the pricing performance of the various
estimators on ``expensive'' options, i.e. whose (observed, not
estimated) price is higher than $1\$$, respectively. 

\begin{table}[tbp]
\caption{Pricing errors on options with price larger than $\$1$ (full dataset)}
\begin{center}
\vspace{2mm}
\parbox{\textwidth}{\footnotesize The table displays descriptive
  statistics on the empirical pricing errors (in percentage points) of
  the various estimators restricted to options in the full dataset
  whose price is larger than $\$1$. Some points of the empirical
  distribution of the pricing error are also reported.  The labels
  used for the different estimators are defined on
  page~\pageref{etichette}. The sample period is January 3, 2012 to
  December 31, 2012. The pricing error is computed on 3222 put options
  and on 2137 call options. \vspace{4mm}}
\footnotesize{{\begin{center} 
 \begin{tabular*}{\textwidth}{@{\extracolsep{\fill}}lccccccc} 
 \toprule 
 \multicolumn{8}{c}{\textsc{Put Options}}\\ 
 \multicolumn{8}{c}{}\\ 
Error & LI & BS& NW & BS-NW & NW\raisebox{-.4ex}{\tiny CV} & BS-NW\raisebox{-.4ex}{\tiny CV} & VG \\  
 \midrule 
Mean       & 7.4& 2.4& 37.9& 20.4& 17.8& 12.3& 47.7\\ 
St. Dev.   & 31.0& 7.4& 71.2& 35.8& 57.4& 33.9& 24.1\\ 
Median     & 2.4& 0.6& 17.6& 9.0& 9.1& 4.0& 47.2\\ 
Min        & 0.0& 0.0& 0.0& 0.0& 0.0& 0.0& 0.0\\ 
Max        & 777.4& 158.0& 894.7& 1051.8& 2170.1& 722.8& 100.0\\ 
 \multicolumn{8}{c}{}\\ 
 \multicolumn{8}{c}{Empirical Distribution}\\ 
 \midrule 
$1\%$      & 25.8& 60.7& 3.8& 8.3& 7.7& 21.1& 0.7\\ 
$5\%$      & 73.2& 88.6& 17.6& 34.5& 34.0& 56.0& 3.6\\ 
$10\%$     & 86.8& 95.6& 32.5& 52.8& 53.4& 73.6& 7.2\\ 
$20\%$     & 94.2& 98.6& 54.7& 70.6& 77.6& 86.0& 15.0\\ 
$25\%$     & 95.5& 98.9& 62.9& 75.5& 83.2& 88.9& 19.4\\ 
$30\%$     & 96.6& 99.1& 68.8& 79.6& 87.4& 91.1& 24.4\\ 
$50\%$     & 98.0& 99.6& 84.2& 90.0& 94.4& 95.3& 54.2\\ 
 \multicolumn{8}{c}{}\\ 
 \multicolumn{8}{c}{}\\ 
 \toprule 
 \multicolumn{8}{c}{\textsc{Call Options}}\\ 
 \multicolumn{8}{c}{}\\ 
Error & LI & BS& NW & BS-NW & NW\raisebox{-.4ex}{\tiny CV} & BS-NW\raisebox{-.4ex}{\tiny CV} & VG \\  
 \midrule 
Mean       & 12.7& 3.4& 43.4& 13.6& 25.1& 6.8& 61.6\\ 
St. Dev.   & 51.5& 13.3& 100.5& 23.3& 63.8& 14.0& 26.4\\ 
Median     & 3.6& 0.8& 21.6& 5.4& 14.2& 3.0& 71.1\\ 
Min        & 0.0& 0.0& 0.0& 0.0& 0.0& 0.0& 0.0\\ 
Max        & 1287.4& 270.5& 3414.3& 355.6& 1437.5& 273.0& 98.9\\ 
 \multicolumn{8}{c}{}\\ 
 \multicolumn{8}{c}{Empirical Distribution}\\ 
 \midrule 
$1\%$      & 20.7& 55.8& 3.1& 12.4& 7.5& 23.6& 1.9\\ 
$5\%$      & 58.4& 85.8& 12.5& 47.5& 25.4& 65.9& 4.7\\ 
$10\%$     & 77.3& 93.8& 22.4& 67.4& 38.9& 83.1& 7.4\\ 
$20\%$     & 88.3& 97.3& 46.4& 81.5& 63.9& 93.4& 12.4\\ 
$25\%$     & 90.7& 98.2& 57.0& 85.7& 71.3& 95.2& 14.5\\ 
$30\%$     & 92.0& 98.7& 63.9& 87.8& 78.1& 96.4& 16.6\\ 
$50\%$     & 95.7& 99.3& 80.4& 93.6& 90.6& 98.6& 27.5\\ 
\bottomrule 
\end{tabular*} 
\end{center}
}}
\end{center}
\label{table6}
\end{table}

The corresponding results are reported in Tables \ref{table6}: the BS
estimator still exhibits a quite strong performance with mean absolute
error around $2$\,-\,$3$\% and more than $85$\% of options priced
within a $5$\% error margin. On the other hand, other estimators show
a somewhat inconsistent performance: the LI estimator for put options
has a mean price error of approximately $7$\%, which increases to over
$12$\% for call options, while a completely symmetrical behavior is
displayed by the BS-NW estimator. The main qualitative conclusions
drawn above are confirmed also in this situation: among the NW-based
methods, only the BS-NW\raisebox{-.4ex}{\scriptsize CV} is in some
cases acceptable, and the simple linear interpolation-based BS
estimator still outperforms (by far) all others. 

Moreover, a comparison of the cumulative distribution
  functions of (the absolute value of) the relative pricing error
  reveals that the BS estimator stochastically dominates all other
  estimators (see Figure~\ref{dom}). This conclusion is also supported
  by the results of a two-sample Kolmogorov-Smirnov test, which
  rejects the hypothesis at confidence level 5\% that the errors of
  \emph{any} two estimators may come from the same distribution.

\begin{figure}[tbp]
\caption{Cumulative distribution functions of pricing errors}
\label{dom}\vspace{2mm}
\parbox{\textwidth}{\footnotesize This figure shows the cumulative distribution functions of pricing errors of the
various estimators, on the basis of the complete dataset. The labels used for the different estimators are defined on page 26.
The sample period is January 3, 2012 to December 31, 2012. The pricing error is computed on 2732 call options (Panel A) and on 4460 put
options (Panel B).}
\newline
\vspace{2mm}
\begin{tabular*}{\textwidth}{@{\extracolsep{\fill}}cc}
Panel A & Panel B  \\
\\
\includegraphics[width=0.5\textwidth]{Dom_call.eps} &
\includegraphics[width=0.5\textwidth]{Dom_put.eps} \\
\label{shock_gdp}
\end{tabular*}
\end{figure}

\subsection{Pricing outside the hull}\label{ssec:fuori}
As seen above, the NW-based as well as the VG estimators can (in
principle) estimate the prices at time $t$ of options whose parameters
$K$ and $T-t$ do not fall within the convex hull of
$(K_k,T_k-t)_{k \in J^0_t}$. Let us recall, however, that NW
non-parametric regression simply produces estimates that are weighted
averages, hence, as already remarked in \S\ref{ssec:NW}, estimates for
parameters falling outside the convex hull of $J^0_t$ should be taken
with extreme care. Results on the empirical pricing error relative to
options whose parameters fall outside the convex hull of $J_t^0$,
$t=1,\ldots,250$, are reported in Tables \ref{table4} and
\ref{table5}: with mean errors over $30$\% and $50$\% or more options
mispriced by at least $20$\%, the BS-NW\raisebox{-.4ex}{\scriptsize CV}
estimator, which also in this case performs better than all other
NW-based as well as VG estimators, could at best be used to get rough
estimates of the correct price.

\begin{table}[tbp]
\caption{Pricing errors outside the convex hull (full dataset)}
\begin{center}
\vspace{2mm}
\parbox{\textwidth}{\footnotesize The table displays descriptive
  statistics on the empirical pricing errors (in percentage points) of
  the various estimators restricted to those options for which the LI
  estimator is not defined, on the basis of the full dataset. Some
  points of the empirical distribution of the pricing error are also
  reported. The labels used for the different estimators are defined
  on page~\pageref{etichette}. The sample period is January 3, 2012 to
  December 31, 2012. The pricing error is computed on 241 put options
  and on 217 call options. \vspace{4mm}}
\footnotesize{{\begin{center} 
 \begin{tabular*}{\textwidth}{@{\extracolsep{\fill}}lccccccc} 
 \toprule 
 \multicolumn{6}{c}{\textsc{Put Options}}\\ 
 \multicolumn{6}{c}{}\\ 
Error &NW & BS-NW & NW\raisebox{-.4ex}{\tiny CV} & BS-NW\raisebox{-.4ex}{\tiny CV} & VG \\  
 \midrule 
Mean       & 996.3& 50.7& 1945.7& 49.8& 62.9\\ 
St. Dev.   & 3381.5& 45.5& 10228.9& 49.3& 45.1\\ 
Median     & 69.0& 45.5& 66.5& 38.7& 50.8\\ 
Min        & 0.0& 0.0& 0.0& 0.0& 0.0\\ 
Max        & 44100.0& 253.4& 128466.9& 380.6& 271.1\\ 
 \multicolumn{6}{c}{}\\ 
 \multicolumn{6}{c}{Empirical Distribution}\\ 
 \midrule 
$1\%$      & 3.7& 7.5& 8.7& 8.7& 0.8\\ 
$5\%$      & 6.2& 19.9& 10.8& 24.9& 4.1\\ 
$10\%$     & 7.1& 34.4& 13.3& 37.8& 13.7\\ 
$20\%$     & 14.1& 41.1& 22.4& 44.4& 22.0\\ 
$25\%$     & 18.7& 44.4& 27.4& 45.6& 24.1\\ 
$30\%$     & 22.4& 46.1& 30.7& 48.1& 29.5\\ 
$50\%$     & 33.2& 51.0& 41.9& 52.3& 49.8\\ 
 \multicolumn{6}{c}{}\\ 
 \multicolumn{6}{c}{}\\ 
 \toprule 
 \multicolumn{6}{c}{\textsc{Call Options}}\\ 
 \multicolumn{6}{c}{}\\ 
Error & NW & BS-NW & NW\raisebox{-.4ex}{\tiny CV} & BS-NW\raisebox{-.4ex}{\tiny CV} & VG \\  
 \midrule 
Mean       & 1527.5& 39.9& 2110.3& 36.4& 53.6\\ 
St. Dev.   & 8394.3& 47.6& 13801.4& 45.7& 57.2\\ 
Median     & 62.2& 11.2& 63.6& 15.5& 47.7\\ 
Min        & 0.3& 0.0& 0.0& 0.0& 0.0\\ 
Max        & 107881.1& 278.1& 153605.4& 278.1& 342.5\\ 
 \multicolumn{6}{c}{}\\ 
 \multicolumn{6}{c}{Empirical Distribution}\\ 
 \midrule 
$1\%$      & 2.3& 19.4& 6.5& 20.3& 18.4\\ 
$5\%$      & 7.4& 39.6& 10.1& 40.1& 23.0\\ 
$10\%$     & 15.2& 48.8& 16.1& 46.1& 29.5\\ 
$20\%$     & 27.6& 55.3& 25.3& 53.9& 37.3\\ 
$25\%$     & 33.6& 57.1& 29.0& 55.3& 40.6\\ 
$30\%$     & 35.9& 59.0& 34.1& 59.0& 43.8\\ 
$50\%$     & 43.8& 63.1& 42.4& 67.7& 51.6\\ 
\bottomrule 
\end{tabular*} 
\end{center}
}}
\end{center}
\label{table4}
\end{table}

\begin{table}[tbp]
\caption{Pricing errors outside the convex hull (trimmed dataset)}
\begin{center}
\vspace{2mm}
\parbox{\textwidth}{\footnotesize The table displays descriptive
  statistics on the empirical pricing errors (in percentage points) of
  the various estimators restricted to those options for which the LI
  estimator is not defined, on the basis of the trimmed dataset. Some
  points of the empirical distribution of the pricing error are also
  reported.  The labels used for the different estimators are defined
  on page~\pageref{etichette}. The sample period is January 3, 2012 to
  December 31, 2012. The pricing error is computed on 251 put options
  and on 225 call options. \vspace{4mm}}
\footnotesize{{\begin{center} 
 \begin{tabular*}{\textwidth}{@{\extracolsep{\fill}}lccccccc} 
 \toprule 
 \multicolumn{6}{c}{\textsc{Put Options}}\\ 
 \multicolumn{6}{c}{}\\  
Error & NW & BS-NW & NW\raisebox{-.4ex}{\tiny CV} & BS-NW\raisebox{-.4ex}{\tiny CV} & VG \\  
 \midrule 
Mean       & 682.9& 51.5& 414.5& 46.1& 58.0\\ 
St. Dev.   & 1404.9& 39.6& 1617.0& 38.9& 48.3\\ 
Median     & 120.8& 59.3& 40.5& 44.8& 43.6\\ 
Min        & 0.0& 0.0& 0.0& 0.0& 0.1\\ 
Max        & 9526.7& 232.1& 15964.3& 232.1& 316.5\\ 
 \multicolumn{6}{c}{}\\ 
 \multicolumn{6}{c}{Empirical Distribution}\\ 
 \midrule 
$1\%$      & 0.8& 5.2& 6.4& 8.0& 1.2\\ 
$5\%$      & 1.6& 17.9& 9.6& 24.3& 7.6\\ 
$10\%$     & 5.2& 27.1& 13.9& 30.3& 15.1\\ 
$20\%$     & 11.2& 35.5& 23.1& 38.6& 23.9\\ 
$25\%$     & 13.5& 37.5& 29.9& 40.6& 26.7\\ 
$30\%$     & 15.5& 39.8& 34.7& 42.2& 32.7\\ 
$50\%$     & 33.1& 45.4& 57.4& 52.6& 55.8\\ 
 \multicolumn{6}{c}{}\\ 
 \multicolumn{6}{c}{}\\ 
 \toprule 
 \multicolumn{6}{c}{\textsc{Call Options}}\\ 
 \multicolumn{6}{c}{}\\ 
Error & NW & BS-NW & NW\raisebox{-.4ex}{\tiny CV} & BS-NW\raisebox{-.4ex}{\tiny CV} & VG \\  
 \midrule 
Mean       & 650.2& 34.3& 321.7& 32.8& 46.3\\ 
St. Dev.   & 1357.2& 42.8& 897.7& 45.0& 45.1\\ 
Median     & 77.9& 15.3& 66.7& 16.2& 38.0\\ 
Min        & 0.2& 0.0& 0.0& 0.0& 0.0\\ 
Max        & 7997.1& 277.8& 6147.3& 280.4& 213.2\\ 
 \multicolumn{6}{c}{}\\ 
 \multicolumn{6}{c}{Empirical Distribution}\\ 
 \midrule 
$1\%$      & 1.8& 19.6& 4.4& 20.0& 20.0\\ 
$5\%$      & 4.4& 34.2& 7.6& 36.4& 25.3\\ 
$10\%$     & 9.3& 44.4& 10.7& 43.1& 31.1\\ 
$20\%$     & 21.8& 54.2& 21.8& 51.1& 37.8\\ 
$25\%$     & 25.8& 56.9& 29.8& 56.9& 41.3\\ 
$30\%$     & 30.7& 59.6& 34.7& 61.3& 44.4\\ 
$50\%$     & 43.6& 69.3& 44.4& 75.6& 59.1\\ 
\bottomrule 
\end{tabular*} 
\end{center}
}}
\end{center}
\label{table5}
\end{table}

A way to extend the domain of definition of estimators based on linear
interpolation is to augment $J_t^0$, for each $t$, with a set of
fictitious parameters and corresponding option prices. In particular,
one can add synthetic options with time to maturity equal to $0$, so
their prices are equal to their payoff. For each $t$, augmenting
$J_t^0$ with $30$ couples of the type $(K_k,0)_{k=1,\ldots,30}$, where
$(K_k)$ are increasing, equally spaced, ranging between the smallest
and the largest strike price of \emph{all} options in $J^t$, we obtain
an ``augmented'' LI estimator. This enlarges the domain of definition
of LI, so that the proportion of options that cannot be prices drops
to $2$\% for put options and to $3.3$\% for call options. The
extension procedure just described can be applied to any data set, as
the fictitious prices are universal. However, since observed prices
tend to be higher than theoretical prices for small times, it does not
seem reasonable to eliminate from the data set options with low price,
which in the large majority of cases also have very short time to
maturity, and then to artificially add options with time to maturity
equal to zero. In other words, the extension procedure should be
applied only, if necessary, to the full data set, and this is what we
do. The empirical results are reported in Table \ref{table8}.
\begin{table}[tbp]
\caption{Linear estimator with fictitious options}
\begin{center}
\vspace{2mm}
\parbox{\textwidth}{\footnotesize The table displays descriptive
  statistics on the empirical pricing errors (in percentage points) of
  the ``augmented'' normalized linear estimator, defined in
  \S\ref{ssec:fuori} and labeled LIB, on the basis of the full
  dataset. For comparison, the first column reports the statistics for
  the normalized linear interpolator (LI). The sample period is
  01/03/2012 to 12/31/2012. \vspace{4mm}}
\footnotesize{{\begin{center} 
 \begin{tabular*}{\textwidth}{@{\extracolsep{\fill}}lcccc} 
 \toprule 
 \multicolumn{5}{c}{}\\ 
&\multicolumn{2}{c}{\textsc{Put Options}} & \multicolumn{2}{c}{\textsc{Call Options}}\\ 
 \cline{2-5} 
 \multicolumn{5}{c}{}\\ 
Error & LI & LIB& LI & LIB\\ 
 \midrule 
Mean          & 11.9& 14.0& 30.1& 30.4\\ 
St. Dev.      & 49.0& 47.5& 418.0& 211.6\\ 
Median        & 3.1& 3.5& 5.0& 5.3\\ 
Min           & 0.0& 0.0& 0.0& 0.0\\ 
Max           & 2027.0& 1167.4& 16075.0& 8860.8\\ 
Priced options& 4460& 4607& 2732& 2851\\ 
\bottomrule 
\end{tabular*} 
\end{center}
}}
\end{center}
\label{table8}
\end{table}

\appendix
\section{Appendix}
\subsection{Sensitivity of implied volatility against dividend}
We place ourselves in the well-known Black-Scholes setting. Let us
denote by $C$ the price of a European call option on a dividend-paying
stock, so that
\[
C = Se^{-qt} \Phi(d_+) - Ke^{-rt} \Phi(d_-),
\]
where
\[
d_+ := \frac{\log S/K + \bigl(r-q+\sigma^2/2\bigr)t}{\sigma\sqrt{t}},
\qquad d_- := d_+ - \sigma\sqrt{t}.
\]
Recalling that
\begin{equation}
  \label{eq:BSid}
  Se^{-qt} \phi(d_+) = Ke^{-rt} \phi(d_-),
\end{equation}
it is immediate to obtain
\[
\frac{\partial C}{\partial q} = -tSe^{-qt}\Phi(d_+).
\]
Moreover, one has
\[
\frac{\partial C}{\partial \sigma} = \sqrt{t} S e^{-qt} \phi(d_+).
\]
In particular, denoting the function $(q,\sigma) \mapsto C$ again just
by $C$ (i.e. treating all other parameters as constants), the implicit
function theorem implies that, for any $q_0$, $\sigma_0$, there exists
a function $\varsigma:\erre_+ \to \erre_+$ (to be interpreted as
implied volatility) such that $\sigma_0=\varsigma(q_0)$,
$C=C(q,\varsigma(q))$, and
\[
\frac{\partial \varsigma}{\partial q}(q_0) = 
- \frac{\partial C}{\partial q}(q_0,\sigma_0) \,
  \frac{\partial C}{\partial \sigma}(q_0,\sigma_0)
= \sqrt{t} \frac{\Phi(d_+)}{\phi(d_+)},
\]
where $d_\pm = d_\pm(q_0,\sigma_0)$.
For the reader's convenience, we are going to prove \eqref{eq:BSid},
which is obviously equivalent to
\[
\frac{S}{K} e^{(r-q)t} = \frac{\phi(d_-)}{\phi(d_+)} 
= \exp\Bigl( \bigl(d_-^2 - d_+^2\bigr)/2 \Bigr),
\]
hence also to
\[
\log(S/K) + (r-q)t = \frac{d_-^2 - d_+^2}{2} 
= \frac12 (d_--d_+) (d_-+d_+).
\]
This identity can be verified by a direct computation, thus showing
the validity of \eqref{eq:BSid}.

\bibliographystyle{amsplain}
\bibliography{references}

\end{document}